%% file: main.tex
\documentclass[preprint]{acmtrans2m}
\usepackage{proof} 
\usepackage{amsmath}
\usepackage{amssymb}
\usepackage{stmaryrd}
\usepackage{pstricks,pst-node,pst-text,pst-3d}
\usepackage[all]{xy}
\usepackage{graphicx}
\input{macro}

\title{An Implicit Characterization of PSPACE}
\author{Marco Gaboardi\\ 
Dipartimento di Scienze dell'Informazione, Universit\`a degli Studi di Bologna - 
Mura Anteo Zamboni 7, 40127 Bologna, Italy,
gaboardi@cs.unibo.it
\and Jean-Yves Marion\\
Nancy-University, ENSMN-INPL, Loria 
B.P. 239, 54506 Vandoeuvre-l\`es-Nancy, France,
Jean-Yves.Marion@loria.fr
\and Simona Ronchi Della Rocca\\ 
Dipartimento di Informatica, Universit\`a degli Studi di Torino - 
Corso Svizzera 185, 10149 Torino, Italy,
ronchi@di.unito.it}
\begin{abstract}
We present a
type system for an extension of lambda calculus with a conditional construction, named $\BSTA$, that characterizes the PSPACE class.
This system is obtained by extending $\STA$, 
a type assignment for lambda-calculus inspired by 
Lafont's  Soft Linear Logic and characterizing the PTIME class.
We extend STA by means of a ground type and terms for booleans and
conditional.
The key issue in the design of the type system is 
to manage the contexts in the rule for conditional
in an additive way. Thanks to this rule, we
are able to program polynomial time Alternating Turing Machines.
From the well-known result APTIME = PSPACE, it follows that $\BSTA$ is complete for PSPACE.\\
Conversely, inspired by the simulation of Alternating Turing machines 
by means of Deterministic Turing machine,
we introduce a call-by-name evaluation machine with two memory devices in order to
evaluate programs in polynomial space.
As far as we know, this is the first characterization of  PSPACE  that is based on lambda calculus
and light logics. 
\end{abstract}
\begin{document} 
 \maketitle
\input{intro}
\input{syntax}
\input{strongNormalization}
\input{strategy}
\input{smallstep}
\input{spaceMeasures}
\input{correctness}
\input{completeness}
\input{conclusion}
\bibliographystyle{acmtrans}
\bibliography{soft}
\end{document} 

%% file: macro.tex
\newtheorem{definition}{Definition}
\newtheorem{theorem}{Theorem}
\newtheorem{property}{Property}
\newtheorem{example}{Example}
\newtheorem{remark}{Remark}
\newtheorem{lemma}{Lemma}

\newtheorem{notation}{Notation}{\bf}{\it}
\newcommand{\lin}{\ensuremath{\multimap}}
\newcommand{\redbeta}{\rightarrow_{\beta} }
\newcommand{\reddelta}{\rightarrow_{\delta} }
\renewcommand{\red}{\rightarrow_{\beta\delta} }
\newcommand{\der}{\vdash }
\newcommand{\derf}{\vdash_{F} }
\newcommand{\pder}{\mathop{\rhd} }
\newcommand{\df}{\doteq}

\newcommand{\0}{{\tt 0}}
\renewcommand{\1}{{\tt 1}}

\newcommand{\lif}[3]{\ {\tt if }\  #1\  {\tt then }\   #2\  {\tt else }\  #3\ }
\newcommand{\ifz}{\ {\tt if }\ }
\newcommand{\thenz}{\ {\tt then }\ }
\newcommand{\elsez}{\ {\tt else }\ }

\newcommand{\BSTA}{\mathrm{STA}_{\bt}}
\newcommand{\STA}{\mathrm{STA}}
\newcommand{\TM}{\mathbf{ATM}}
\newcommand{\bt}{\mathbf{B}}
\newcommand{\T}{\mathcal{T}}
\newcommand{\A}{\Gamma}
\newcommand{\B}{\Delta}

\newcommand{\ev}{\Downarrow}
\newcommand{\FV}{\mathrm{FV}}
\newcommand{\FTV}{\mathrm{FTV}}

\newcommand{\dom}{\mathrm{dom}}
\newcommand{\prog}{\mathcal{P}}
\newcommand{\pathi}{\tt path}

\newcommand{\M}{{\tt M}}
\newcommand{\N}{{\tt N}}
\newcommand{\R}{{\tt R}}
\newcommand{\Q}{{\tt Q}}
\newcommand{\V}{{\tt V}}

\newcommand{\x}{{\tt x}}
\newcommand{\y}{{\tt y}}

\newcommand{\letp}[4]{\mathtt{let }\ #1\ \mathtt{ be }\  #2,#3\ \mathtt{ in }\ #4}

\newcommand{\rk}{{\tt rk}}


%% file: intro.tex
\section{Introduction}

The argument of this paper fits in the so called Implicit Computational Complexity area, whose aim is to provide
complexity control through language restrictions, without using explicit machine models or external measures.
In this setting, we are interested in the design of programming languages with bounded computational complexity.
We want to use a ML-like approach, so having a $\lambda$-calculus like language, and a 
type assignment system for it, where the types guarantee, besides the functional correctness,
also complexity properties. So, types can be used in a static way in order to check the correct 
behaviour of the programs, also with respect to the resource usage. 
According to these lines, we design in this paper a language correct and complete with respect to PSPACE.
Namely, we supply, besides the calculus, a type assignment system and an evaluation machine, and we prove 
that well typed programs can be evaluated by the machine in polynomial space, and moreover that all
decision functions computable in polynomial space can be coded by well typed programs.\\
\medskip

\textbf{Light Logics and type systems}
Languages characterizing complexity classes through type assignment systems for $\lambda$-calculus are already present in the
literature, but they are quite all related to time complexity. 
The key idea is to use as types the formulae of the light logics, which characterize some classes of time complexity:
Light Linear Logic (LLL) of Girard \cite{Girard98ic}, and Soft Linear Logic (SLL) of Lafont \cite{Lafont04}
characterize polynomial time, while Elementary Linear Logic (EAL) characterizes 
elementary time. The characterization is based on the fact that cut-elimination on proofs in these logics is performed 
in a number of steps which depends in a polynomial or elementary way from the initial size of the proof (while the
degree of the proof, i.e., the nesting of exponential rules, is fixed). 
Moreover, the size of each proof in the cut elimination process can be bound by a  polynomial or an elementary function in the initial size of the proof, respectively.
In addition, all these logics are 
also complete with respect to the related complexity class, using proof-nets for coding functions.\\
The good properties of such logics have been fruitfully used in order to design type assignment systems for
$\lambda$-calculus which are correct and complete with respect to the polynomial or elementary time complexity bound.
Namely, every well typed term $\beta$-reduces to normal form in a number of steps that depends 
in a polynomial or elementary way from its size, and moreover all functions with the corresponding complexity
are representable by a well typed term. Examples of polynomial type assignment systems are in  
\cite{Baillot04lics,journals/iandc/BaillotT09} and ~\cite{GaboardiRonchi07csl,GaboardiRonchi09igpl}, based respectively on LAL (an affine variant of LLL
designed by Asperti and Roversi ~\cite{Asperti02tocl}) and on SLL.
Moreover, an example of an elementary type assignment 
system is in 
\cite{Cop-DLag-Ron:EALCBV-05,DBLP:journals/lmcs/CoppolaLR08}.\\
\medskip

\textbf{Contribution} In order to use a similar approach for measuring space complexity, since there is no previous 
logical characterization of PSPACE from which we can start, we exploit the fact that polynomial space
computations coincide with polynomial time alternating Turing machine computations
(APTIME). In particular, by the results in ~\cite{Sav70} and  ~\cite{ChKS81}, it follows
$$
  \text{PSPACE} = \text{NPSPACE} = \text{APTIME}
$$
So, we start from the type assignment system  STA for $\lambda$-calculus introduced in 
\cite{GaboardiRonchi07csl}. It is based on SLL, in the sense that in STA both types are a proper subset
of SLL formulae, and type assignment derivations correspond, through the Curry-Howard isomorphism,
to a proper subset of SLL derivations. STA is correct and complete (in the sense said before) 
with respect to polynomial time computations. 

Then we design the language
$\Lambda_{\mathcal{B}}$, which is an extension of $\lambda$-calculus with
two boolean constants and a conditional constructor, and we supply it by a type assignment system ($\BSTA$), 
where the types are STA types plus a constant type $\bt$ for booleans, and rules for conditional.
In particular, the elimination rule for conditional is the following:
 $$
\infer[(\bt E)]{\Gamma \der \lif{{\tt M}}{{\tt N_{\0}}}{{\tt N_{\1}}}:A}{\Gamma\der 
{\tt M}:\bt & \A\der {\tt N_{\0}}:A  &\A\der {\tt N_{\1}}:A}
$$
In this rule, contexts are managed in an additive way, that is
with free contractions. From a computational point of
view, this intuitively means that a computation can repeatedly fork
into subcomputations and the result is obtain by a backward
computation from all subcomputation results. \\
While the time complexity result for STA is not related to a particular evaluation strategy,
here, for characterizing space complexity, the evaluation should be done carefully. Indeed, an uncontrolled 
evaluation can construct exponential size terms. So
we define a call-by-name SOS evaluation machine, $\mathrm{K}_{\mathcal{B}}^{\mathcal{C}}$, inspired by Krivine's machine
\cite{Krivine07} for $\lambda $-calculus, where substitutions are made only on head variables. 
This machine is equipped with two memory devices, and the space used by it
is proved to be the dimension of its maximal configuration. The proof is made through the design of an
equivalent small-step machine.
Then we prove that, if $\mathrm{K}_{\mathcal{B}}^{\mathcal{C}}$ takes a program (i.e., a closed 
term well typed with a constant type) as input, then
the size of each configuration is polynomially bounded 
in the
size of the input. So every program is evaluated by the machine in 
polynomial space.
Conversely, we encode every polynomial time alternating Turing machine by a
program well typed in $\BSTA$. The simulation relies on a higher order
representation of a parameter substitution recurrence schema inspired by the one in~ \cite{conf/csl/Marion94}. \\
\medskip

\textbf{Related works}
The present work extends the preliminary results that have been presented to POPL '08 \cite{GaboardiMarionRonchi08popl}.
The system $\BSTA$ is the first characterization of PSPACE through a type assignment system in the light logics setting.
A proposal for a similar characterization has been made by Terui 
\cite {Terui00}, but the work has never been completed.

The characterization presented here is strongly based on the additive rule $(\bt E)$ presented above.
The key role played by this rule in the characterization of 
the PSPACE class has been independently suggested by Hofmann in the context of non-size-increasing computations \cite{journals/iandc/Hofmann03}.
There, the author showed that by adding to his LFPL language
a form of \emph{restricted duplication} one can encode the ``quantified 
boolean formulas problem'' and recover exactly the
behaviour of the rule $(\bt E)$.
Besides the difference in the setting where our study is developed with respect
to the Hofmann one, our work improves on this 
in the fact that we give a concrete syntactical proof of PSPACE soundness for programs by means of an evaluation machine while Hofmann PSPACE soundness relies on a semantic argument that
hides the technical difficulties that one needs to deal with in the 
evaluation of programs. Moreover, we here give a PSPACE completeness result based on the definability of all polynomial time Alternating Turing Machines.

In our characterization we make use of boolean constants in order to have a fine control of the space needed to evaluate  programs. A
use of constants similar in spirit to the present one 
has been also employed by the second author in 
\cite{conf/tlca/LeivantM93}, in order to give a characterization of the PTIME class.

There are several other implicit characterizations of polynomial space
computations using principles that differ from the ones explored in this paper. The characterizations in ~\cite{conf/csl/Marion94,1997:tapsoft:leivant} and 
\cite{conf/dagstuhl/Oitavem01,DBLP:journals/mlq/Oitavem08} are based on ramified recursions
over binary words. In finite model theory,
PSPACE is captured by first order queries with a partial fixed point
operator \cite{Var82,AV89}. The reader may consult the recent
book~\cite{FMTA}. Finally there are some algebraic characterizations like the
one~\cite{Go92b} or~\cite{Jones00} but which are, in essence, over
finite domains.

Apart from the class PSPACE, the light logic principles have been 
used to characterize other interesting complexity classes.  
In \cite{conf/tlca/Maurel03} and \cite{GaboardiMarionRonchi08lsfa} 
an explicit sum rule to deal with 
non deterministic computation has been studied in the setting of
Light Linear Logic and Soft Linear Logic, respectively. Both these
works give implicit characterizations of the class NPTIME.
Another important  work in this direction is the one 
in \cite{Schopp07lics} where 
a logical system characterizing
logarithmic space computations is defined, the  
Stratified Bounded Affine Logic (SBAL). 
Interestingly, the logarithmic space 
soundness for SBAL is proved in an interactive way by means of a geometry of 
interaction algorithm considering only proofs of certain sequents to represent
the functions computable in logarithmic space. This idea was already present in the previous work \cite{conf/csl/Schopp06} of the same author and it has been further
explored in the recent work \cite{conf/esop/LagoS10}.\\

\medskip

\textbf{Outline of the paper}
In Section 2 the system $\BSTA$ is introduced and the proofs of
subject reduction and strong normalization properties are given. In Section 3 the operational semantics
of $\BSTA$ program is defined, through two equivalent abstract evaluation machines.
In Section 4 we show that $\BSTA$ programs can be executed in polynomial
space. In Section 5 the completeness for PSPACE is proved.
Section 6 contains some conclusions.


%% file: syntax.tex
\section{Soft Type Assignment system with Booleans}
In this section we present the paradigmatic language 
$\Lambda_{\mathcal{B}}$ and a type assignment for it, $\BSTA$, and we will prove that
$\BSTA$ enjoys the properties of subject reduction and strong normalization.
$\Lambda_{\mathcal{B}}$ is an extension of the $\lambda$-calculus with boolean constants 
$\0,\1$ and an $\ifz$ constructor. 
$\BSTA$ is an extension of the type system $\STA$ for 
$\lambda$-calculus introduced in \cite{GaboardiRonchi07csl},
which assigns to $\lambda$-terms a proper subset of 
formulae of Lafont's Soft Linear Logic \cite{Lafont04}, and it is
correct and complete for polynomial time
computations. 
\begin{definition}[($\Lambda_{\mathcal{B}}$)] \label{def:termGrammar} \ 
\begin{enumerate}
\item
The set $\Lambda_{\mathcal{B}}$ of \emph{terms} is defined by the 
following grammar:
$$
\M::= \x\mid \0\mid \1\mid \lambda \x.\M\mid \M\M\mid 
\lif{\M}{\M}{\M}
$$
where $\x$ ranges over a countable set of variables and 
$\mathcal{B}=\{\0,\1\}$ is the set of \emph{booleans}.

\item The \emph{reduction relation} $\red\subseteq \Lambda_{\mathcal{B}}\times\Lambda_{\mathcal{B}}$ is the contextual closure of the following rules:
$$
(\lambda \x.\M)\N\redbeta \M[\N/\x]
$$ 
$$
\lif{\0}{\M}{\N}\reddelta \M 
$$
$$
\lif{\1}{\M}{\N}\reddelta \N
$$
$\red^{+}$ denotes the transitive closure of $\red$ and 
$\red^{*}$ denotes the reflexive closure of $\red^{+}$.
\item The \emph{size}  of a term $\tt M$ is denoted as $|\tt M|$ and is defined inductively as
  $${|\tt x |=|0|=|1|}=1 \qquad {\tt |\lambda x.M|=|M|+}1\qquad 
{\tt |MN|=|M|+|N|}
$$
$$
|\lif{\M}{{\tt N_0}}{{\tt N_1}}|=|{\tt M}|+|{\tt N_0}|+|{\tt N_1}|+1
$$
\end{enumerate}
\end{definition}
Note that we use the term $\0$ to denote ``true'' and  
the term $\1$ to denote ``false''.
\begin{notation}
Terms are denoted by $\tt M,N,V,P$. In order to avoid unnecessary parenthesis, we use the Barendregt convention,
so abstraction associates on the left and applications associates on the right. Moreover
$\tt \lambda x y.\tt M$ stands for $\tt \lambda x.\lambda y.\tt M$.
As usual terms are considered up to $\alpha$-equivalence, namely a bound 
variable can be renamed provided no free variable is captured. 
Moreover, $\M[\N/\x]$ denotes the capture-free substitution of all free 
occurrences of $\x$ in $\M$ by $\N$,
$\FV(\M)$ denotes the set of free variables of $\M$ and $n_o(\x,\M)$ denotes the number of free occurrences of the 
variable $\x$ in $\M$.
\end{notation}
In the sequel we will be interested only in typable terms.
\begin{definition}[($\BSTA $)]\label{def:typeGrammar} \ 
\begin{enumerate}
\item
The set $\T_{\bt} $ of \emph{types} is defined as follows:
\begin{center}
\begin{tabular}{ll}
$A::= \bt \mid \alpha \mid \sigma \lin A  \mid \forall \alpha.A $ & (Linear Types)\\
$\sigma ::= A\hspace{0.7mm} \mid !\sigma$
\end{tabular}
\end{center}
where $\alpha$ ranges over a countable set of type variables and 
$\bt$ is the only \emph{ground} type. 
\item
A \emph{context} is a set of assumptions of the shape $\x:\sigma$,
where all variables are different. 
We use $\Gamma, \Delta$ to denote contexts.
\item
The system $\BSTA$ proves judgments of the shape
$\Gamma \vdash \M:\sigma$
where $\Gamma$ is a context, $\M$ is a term, and $\sigma$ is a  type.
The rules are given in Table \ref{systemDefinition}.
\end{enumerate}   
\end{definition}
\begin{notation}
Type variables are denoted by $\alpha, \beta$, linear types by $A,B,C$, and 
types by $\sigma,
\tau, \mu$. The symbol  $\equiv $ denotes the syntactical equality both for types and terms
(modulo renaming of bound variables). 
As usual $\lin$ associates to the right and has precedence on $\forall$, while $!$ has precedence on everything else.
The notation $\sigma[A/\alpha]$ stands for the usual capture free 
substitution in $\sigma$ of all occurrences of the type variable
$\alpha$ by the linear type $A$. 
We use $\dom(\Gamma)$ and $\FTV(\Gamma)$ to denote respectively
the sets of variables and of free type 
variables  that occur in the assumptions of the context $\Gamma$.
The notation $\Gamma\#\Delta$ stands for  $\dom(\Gamma)\cap \dom(\Delta)=\emptyset$.
Derivations are denoted by $\Pi, \Sigma, \Theta$. $\Pi\pder  \Gamma \vdash \M:\sigma$ denotes a derivation $\Pi$
with conclusion $\Gamma \vdash \M:\sigma$.  We let $\vdash \M:\sigma$ 
abbreviate 
$\emptyset \vdash \M:\sigma $. 
As usual, $\forall{\vec{\alpha}}.A$ is an abbreviation for $\forall \alpha_{1}....\forall\alpha_{m}.A$,
and 
$!^{n}\sigma$ is an abbreviation for $!...!\sigma$ $n$-times  ($m,n \geq 0$). 
\end{notation}
We stress that each type is of the shape $!^{n}\forall{\vec{\alpha}}.A$.
The type assignment system $\BSTA$ is obtained form $\STA$ just by 
adding the rules
for dealing with the $\ifz$ constructor. Note that the rule $(\bt E)$ has an additive treatment of the contexts,
and so 
contraction is free, while all other rules are multiplicative. Moreover $\BSTA$ is  affine, since the weakening is free,
so it enjoys the following properties.

\begin{table}
\begin{center}
\begin{tabular}{|c|}
\hline
\\
$
\begin{array}{lclclcl}
\text{(Linear Types)}\quad  A, B&:=& \bt\ |\ \alpha\ |\ \sigma\multimap A\ |\ \forall \alpha.A&\qquad\quad&\text{(Types)}\quad\sigma,\tau&:=& A \ |\ !\sigma
\end{array}
$\\
\\
\hline
\\
$
\infer[(Ax)]{{\tt x}:A\der {\tt x}:A}{}
\qquad 
\infer[(\bt_{\0} I)]{\der \0:\bt}{}
\qquad
\infer[(\bt_{\1} I)]{\der \1:\bt}{}
\qquad
\infer[(w)]{\A,\x:A\der {\tt M}:\sigma}{\A \der {\tt M}:\sigma  }
$\\
\\
$
\qquad
\infer[(\lin I)]{\A\der {\tt \lambda x.M}:\sigma\lin A}{\A, {\tt x}:\sigma\der {\tt M}:A }
\qquad 
\infer[(\lin E)]{\A,\B\der {\tt MN}:A}{\A\der {\tt M}:\sigma\lin A & \B \der {\tt N}:\sigma & \A \# \B}
$
\\
\\
$
\infer[(m)]{\A, {\tt x}:!\sigma\der {\tt M}[{\tt x}/{\tt x}_1,\cdots,{\tt x}/{\tt x}_n ]:\tau}{\A,{\tt x}_1:\sigma,\ldots,{\tt x}_n:\sigma\der {\tt M}:\tau}
\qquad
\infer[(sp)]{!\A\der {\tt M}:!\sigma}{\A\der {\tt M}:\sigma}
\qquad
\infer[(\forall E)]{\A\der {\tt M}:B[A/\alpha]}
{\A\der {\tt M}:\forall \alpha .B }
$
\\
\\
$
\infer[(\bt E)]{\Gamma \der \lif{{\tt M}}{{\tt N_{\0}}}{{\tt N_{\1}}}:A}{\Gamma\der {\tt M}:\bt & \Gamma\der {\tt N_{\0}}:A  &\Gamma\der {\tt N_{\1}}:A}
\qquad
\infer[(\forall I)]{\A\der {\tt M}:\forall \alpha .A}
{\A\der {\tt M}:A & \alpha\notin \FTV(\A)}

$
\\
\\
\hline
\end{tabular}
\end{center}
\caption{The Soft Type Assignment system with Booleans}
\label{systemDefinition}
\end{table}
\begin{lemma}[Free variable lemma] \
\begin{enumerate}
\item $\A\der \M:\sigma$ implies $\FV(\M)\subseteq \dom(\A)$.
\item $\A\der \M:\sigma, \B\subseteq \A$ and $\FV(\M)\subseteq \dom(\B)$
imply $\B\der \M:\sigma$.
\item $\A\der \M:\sigma, \A\subseteq \B$ implies $\B\der \M:\sigma$.
\end{enumerate}
\end{lemma}
\begin{proof}
All the three points can be easily proved by induction on the derivation proving $\A\der \M:\sigma$. 
\end{proof}
Moreover, the following property holds:
\begin{lemma}
\label{lem:linearVarnot}
  $\Gamma, \x:A \der \M:!\sigma$ implies $\x \not\in \FV(\M)$.
\end{lemma}
\begin{proof}
Easy, by induction on the derivation proving $\Gamma, \x:A \der \M: !\sigma$ noticing that the only way to have a modal conclusion is by using the $(sp)$ rule.
\end{proof}
In what follows, we  will need to talk about proofs modulo some simple operations.
\begin{definition}
Let $\Pi$ and $\Pi'$ be two derivations in $\BSTA$, proving the same conclusion. Then,  
$\Pi \leadsto \Pi'$ denotes the fact that $\Pi'$ is obtained from $\Pi$
by commuting or deleting some rules or by inserting some applications
of the rule $(w)$.
\end{definition}

The system $\BSTA$ is not syntax directed, but the Generation Lemma shows that we can modify the derivations, using just
commutation and erasing of rules, in order to connect the shape of a term with the shape of its typings.\\ 

\begin{lemma}[Generation lemma]\
\label{lem:generation}
\begin{enumerate}
\item $\Pi\pder \Gamma \vdash \lambda \x.\M:\forall \alpha.A$ implies
  there is $\Pi'$, proving the same conclusion as $\Pi$ and ending with an application of rule $(\forall I)$, such that 
$\Pi \leadsto \Pi'$.
\item $\Pi\pder \Gamma \vdash \lambda \x.\M: \sigma \lin A$ implies
  there is $\Pi'$, proving the same conclusion as $\Pi$ and ending with an application of rule $(\lin I)$,
  such that $\Pi\leadsto \Pi'$.
\item $\Pi\pder \Gamma \vdash \M:!\sigma$ implies there is $\Pi'$,
proving the same conclusion as $\Pi$,
  such that $\Pi \leadsto \Pi'$
and $\Pi'$ consists of a subderivation, 
ending with the rule ($sp$) proving 
$!\Gamma' \vdash \M: !\sigma$, followed by a sequence of rules ($w$) and/or ($m$)  dealing with variables not occurring in $\M$.
\item $\Pi \pder !\Gamma \vdash \M:!\sigma$ implies there is $\Pi'$, proving the same 
conclusion as $\Pi$ and ending with an application of rule $(sp)$,
  such that $\Pi\leadsto \Pi'$.
\end{enumerate}
\end{lemma}
\begin{proof}\begin{enumerate}
\item
By induction on $\Pi$. If the last rule of $\Pi$ is $(\forall I)$ then the conclusion follows immediately. 
Otherwise consider the case 
$\lambda\y.\M\equiv\lambda\y .\N[\x/\x_1,\cdots, \x/\x_n]$ and
$\Pi$ ends as:
$$
\infer[(m) ]{\Gamma,\x:!\sigma  \der \lambda \y.\N[\x/\x_1,\cdots, \x/\x_n]:\forall \alpha.A }
{\Sigma\pder \Gamma,\x_1:\sigma,\ldots,\x_n:\sigma \der \lambda \y.\N:\forall \alpha.A  }
$$
By induction hypothesis $\Sigma\leadsto \Sigma'$ ending as:
$$
\infer[(\forall I) ]{\Gamma,\x_1:\sigma,\ldots,\x_n:\sigma \der \lambda \y.\N:\forall\alpha.A  }
{\Sigma_1\pder\Gamma,\x_1:\sigma,\ldots,\x_n:\sigma \der \lambda \y.\N:A}
$$
Then, the  desired $\Pi'$ is:
$$
\infer[(\forall I)]{\Gamma,\x:!\sigma  \der \lambda \y.\N[\x/\x_1,\cdots, \x/\x_n]:\forall \alpha.A }{
\infer[(m) ]{\Gamma,\x:!\sigma  \der \lambda \y.\N[\x/\x_1,\cdots, \x/\x_n]:A }
{\Sigma_1\pder \Gamma,\x_1:\sigma,\ldots,\x_n:\sigma \der \lambda \y.\N:A  }}
$$
The cases where $\Pi$ ends either by $(\forall E)$ or $(w)$ rule are easier. 
The other cases are not possible.
\item Similar to the proof of the previous point of this lemma.
\item 
By induction on $\Pi$. In the case the last rule of $\Pi$ is $(sp)$, 
the proof is obvious. The case where the last rule of $\Pi$
is $(w)$ follows directly by induction hypothesis.\\
Consider the case where $\M\equiv \N[\x/\x_{1},...,\x/\x_{n}]$ and the 
last rule is:
$$
\infer[(m)]{ \Delta,\x:!\tau \der \N[\x/\x_{1},...,\x/\x_{n}]:!\sigma}{\Sigma\pder\Delta,\x_{1}:\tau,...,\x_{n}:\tau \der \N:!\sigma}
$$
In the case $\x_1,\ldots,\x_n\notin \FV(\N)$ the conclusion follows 
immediately.
Otherwise, by induction hypothesis $\Sigma\leadsto \Sigma_{1}$, where $\Sigma_{1}$ is 
composed by a subderivation $\Theta$ ending with a rule ($sp$)  proving
$!\Delta_1 \der {\tt N}: !\sigma$, 
followed by a sequence $\delta$ 
of rules ($w$) or ($m$), dealing with variables not 
occurring in $\tt N$. 
Note that for each  $\x_i$ with $1\leq i\leq n$ such that 
$\x_{i}\in \FV(\N)$, necessarily $\x_i:\tau'\in \Delta_1$ and $\tau=!\tau'$.
Let $\Delta_2$ be the context $\Delta_1-\{\x_1:\tau',\ldots,\x_n:\tau'\}$,
then the conclusion follows by the derivation:
$$
\infer[(sp)]{!\Delta_2,\x:!\tau\der \N[\x/\x_1,\cdots,\x/\x_n]:!\sigma}
{\infer[(m)]{\Delta_2,\x:!\tau' \der \N[\x/\x_1,\cdots,\x/\x_n]:\sigma}{\Delta_2, \x_1:\tau',\ldots,\x_n:\tau' \der \N:\sigma}}
$$
followed by a sequence of rules $(w)$ 
recovering the context $\Delta$ from
the context $\Delta_2$.
The other cases are not possible.
\item
By induction on $\Pi$. In the case the last rule of $\Pi$ is $(sp)$, 
the proof is obvious. The only other possible case is when the last rule 
is $(m)$. Consider the case where $\M\equiv \N[\x/\x_{1},...,\x/\x_{n}]$ and $\Pi$ ends as follows:
$$
\infer[(m)]{ !\Delta,\x:!\tau \der \N[\x/\x_{1},...,\x/\x_{n}]:!\sigma}{\Sigma\pder!\Delta,\x_{1}:\tau,...,\x_{n}:\tau \der \N:!\sigma}
$$
If $\tau \equiv !\tau'$,
by induction hypothesis $\Sigma \leadsto \Sigma_1$, where $\Sigma_1$ ends as:
$$\infer[(sp)]{!\Delta, \x_{1}: !\tau',...,\x_{n}:!\tau' \der \N:!\sigma}{\Theta\pder \Delta, \x_{1}: \tau',...,\x_{n}:\tau' \der \N:
\sigma}$$
So the desired derivation $\Pi'$ is $\Theta$, 
followed by a rule $(m)$ and a rule $(sp)$.
In the case $\tau$ is linear, by Lemma \ref{lem:linearVarnot}, $\x_{i}\not \in \FV(\N)$ 
for each $1\leq i \leq n$.
Moreover  by the previous point of this lemma, $\Sigma$ 
can be rewritten as:
$$\infer[(sp)]{!\Delta_1 \der \N:!\sigma}{\Sigma_1\pder \Delta_1 \der \N:\sigma}$$
followed by a sequence $\delta$ of rules, all dealing with variables not 
occurring in $\N$. So $\delta$ needs to contain some rules introducing the 
variables $\x_{1},...,\x_{n}$. Let $\delta'$ be the sequence of rules
obtained from $\delta$ by erasing such rules, and inserting a $(w)$ rule
 introducing the variable $\x$.
The desired derivation $\Pi'$ is $\Sigma_1$ followed by $\delta'$, followed by ($sp$).
\end{enumerate}
\end{proof}
\subsection{Subject reduction}
In order to prove subject reduction, we need to prove before that the system enjoys the property of
substitution. 
This last property cannot be proved in a standard way, since the linearity of the axioms and the fact that
the rule $(m)$ renames some variables both in the subject and in the context. 
So, in order to prove that $\A,\x:\mu \der \M:\sigma$ and $\B \der \N:\mu$ 
($\A\#\B$) implies
$\A,\B \der \M[\N/\x]:\sigma$, we need to consider 
all the axioms introducing variables which will be renamed as $\x$ in the derivation itself. We need to replace each of them
by a disjoint copy of the derivation proving $\B \der \N:\mu$, and finally to apply a suitable numbers of $(m)$ rules.
In order to formalize this procedure we need to introduce the notion of height of a variable in a derivation.
\begin{definition}\label{def:height}
Let $\Pi\pder \Gamma, \x:\tau \der \M:\sigma$.
The \emph{height} of $\x$ in $\Pi$ is 
inductively defined as follows:
\begin{itemize}
\item if the last rule of $\Pi$ is:
 $$\infer[(Ax)]{\x:A \der \x:A}{}\quad  \textrm{or}\quad  \infer[(w)]{\Gamma', \x:A \der \N:\sigma}{\Gamma'\der \N:\sigma}
$$
then the height of $\x$ in $\Pi$ is $0$.
\item if the last rule of $\Pi$ is:
 $$\infer[(m)]{\Gamma', \x:!\tau \der \N[\x/\x_{1},...,\x/\x_{k}]:\sigma}{\Sigma\pder \Gamma', \x_{1}:\tau,\ldots
,\x_{k}:\tau \der \N:\sigma}$$ 
then 
the height of $\x$ in $\Pi$ is the maximum between the heights of $\x_i$ in 
$\Sigma$  for $1\leq i\leq k$ plus one.
\item If $\x:\tau\in\Gamma$ and  the last rule of $\Pi$ 
is$$
\infer{\Gamma \der \lif{{\tt M}}{{\tt N_{\0}}}{{\tt N_{\1}}}:A}{\Sigma\pder \Gamma\der {\tt M}:\bt & \Theta_{\0}\pder \Gamma\der {\tt N_{\0}}:A  &\Theta_{\1}\pder \Gamma\der {\tt N_{\1}}:A}$$
Then 
the height of $\x$ in $\Pi$ 
is the maximum  between the heights of  $\x$ in $\Sigma,\Theta_{\0}$ and $\Theta_{\1}$
respectively, plus one.
\item In every other case there is only one assumption with subject 
$\x$
both in the context of the conclusion of the rule
and in the context of one of its premises $\Sigma$.
Then the height of $\x$ in $\Pi$ is equal to the height of $\x$ in $\Sigma$
plus one.
\end{itemize}
\end{definition}
We can now prove the substitution lemma.
\begin{lemma}[Substitution lemma]$ $\\
\label{lem:substitution}
Let $\A,\x:\mu \der \M:\sigma$ and $\B \der \N:\mu$ 
such that $\A\#\B$. Then
$$\A,\B \der \M[\N/\x]:\sigma$$
\end{lemma}
\begin{proof}
Let $\Pi$ and $\Sigma$ be the derivations proving respectively $\A,\x:\mu \der \M:\sigma$ and $\B \der \N:\mu$.
By induction on the height of $\x$ in $\Pi$. 
Base cases $(Ax)$ and $(w)$ are trivial. The cases where $\Pi$ ends either by 
$(\lin I),(\forall I), (\forall E)$ or $(\lin E)$ follow 
directly from the induction hypothesis.\\
Let $\Pi$ ends by $(sp)$ rule with premise 
$\Pi'\pder \Gamma', \x:\mu' \der \M:\sigma'$. Then 
by Lemma \ref{lem:generation}.3, $\Sigma\leadsto \Sigma ''$ which is composed 
by 
a subderivation ending with an $(sp)$ rule with premise $\Sigma'\pder \B'\der \N:\mu' $ followed by a sequence of rules $(w)$ and/or $(m)$. By induction 
hypothesis we have a derivation $\Theta'\pder \A',\B'\der \M[\N/\x]:\sigma' $.
By applying the rule $(sp)$ and the sequence of $(w)$ and/or $(m)$ rules we obtain 
$\Theta\pder \A,\B\der \M[\N/\x]:\sigma$.\\
Consider the case $\Pi$ ends by:
$$
\infer[(\bt E)]{\A,\x:\mu \der \lif{{\M_{0} }}{{\M_{1}}}{{\M_{2}}}:A}{\Pi_{0}\pder \A,\x:\mu\der {\M_{0}}:\bt &  \Pi_{1}\pder \A,\x:\mu\der {\M_{1}}:A  &\Pi_{2}\pder \A,\x:\mu\der {\M_{2}}:A}
$$
Then by the induction hypothesis there are derivations
$\Theta_0\pder \Gamma,\Delta \der\M_{0}[\N/\x]:\bt$, 
$\Theta_1\pder \Gamma,\Delta \der\M_{1}[\N/\x]:A$ and 
$\Theta_2\pder \Gamma,\Delta \der\M_{2}[\N/\x]:A$.
By applying a $(\bt E)$ rule we obtain a derivation $\Theta$ with conclusion:
$$\Gamma,\B \der \lif{{\M_{0}[\N/\x] }}{{\M_{1}[\N/\x]}}{{\M_{2}[\N/\x]}}:A$$
Consider the case $\Pi$ ends by: 
$$
\infer[(m)]{\A, {\tt x}:!\mu'\der {\tt M}[{\tt x}/{\tt x}_1,\cdots,{\tt x}/{\tt x}_m ]:\sigma}{\Pi'\pder \A,{\tt x}_1:\mu',\ldots,{\tt x}_m:\mu'\der {\tt M}:\sigma}
$$
By Lemma \ref{lem:generation}.3 $\Sigma\leadsto \Sigma ''$ ending by an 
$(sp)$ rule with premise $\Sigma'\pder \B'\der \N:\mu' $ followed
by a sequence of  rules $(w)$ and/or $(m)$. 
Consider fresh copies of the derivation
$\Sigma'$ i.e. $\Sigma'_{j}\pder \B_{j}'\der \N_{j}:\mu' $
 where $\N_j$ and $\B_j'$ 
 are fresh copies of $\N$ and $\B'$ ($1 \leq j \leq m$).\\
Let $\x_i$ be such that its height is maximal between the heights 
of all $\x_j$ ($1\leq j\leq m$).
By induction hypothesis there is a  derivation: 
$$\Theta_i\pder \A,\x_1:\mu',\ldots,\x_{i-1}:\mu',\x_{i+1}:\mu',\ldots,\x_m:\mu',\B'_i\der \M[\N_i/\x_i]:\sigma $$
Then, we can repeatedly 
apply induction  hypothesis to obtain a derivation
$\Theta'\pder \A,\B'_1,\ldots ,\B'_m\der \M[\N_1/\x_1,\cdots,\N_m/\x_m]:\sigma$.
Finally
 by applying  repeatedly  the rules $(m)$ and $(w)$ 
the conclusion follows.
\end{proof}
We can finally prove the main property of this section.
\begin{lemma}[Subject Reduction]\label{lem:subjectReduction}$ $\\
Let $ \A\der \M:\sigma$ and $\M\red \N$. Then,
$\A\der \N:\sigma$.
\end{lemma}
\begin{proof}
By induction on the derivation $\Theta\pder \A\der \M:\sigma$.
Consider the case of a $\reddelta$ reduction. Without
loss of generality we can consider only the case $\Theta$ ends as:
$$
\infer[(\bt E)]{\A \der \lif{{\tt b}}{{\M_{\0}}}{{\M_{\1}}}:A}{\Pi\pder \A\der {\tt b}:\bt &  \Pi_{\0}\pder \A\der {\M_{\0}}:A  &\Pi_{\1}\pder \A\der {\M_{\1}}:A}
$$
where ${\tt b}$ is either $\0$ or $\1$. The others follow directly 
by induction hypothesis. 
If $\tt b\equiv \0$ then  $\lif{{\tt b}}{{\M_{\0}}}{{\M_{\1}}}\reddelta 
\M_{\0}$ and since $\Pi_{\0}\pder \A\der {\M_{\0}}:A$, the conclusion
follows.
Analogously if $\tt b\equiv \1$ then  $\lif{{\tt b}}{{\M_{\0}}}{{\M_{\1}}}\reddelta 
\M_{\1}$ and since $\Pi_{\1}\pder \A\der {\M_{\1}}:A$, the conclusion
follows.\\
Now consider the case of a $\redbeta$ reduction.
Without loss of generality we can consider only the case $\Theta$ ends
as:
$$
\infer[(\lin E)]{\A_1,\A_2 \der (\lambda \x.\M)\N:A}{
\Pi\pder \A_1\der \lambda \x.\M:\sigma\lin A & \Sigma\pder \A_2\der \N:\sigma
}
$$
where $\Gamma= \Gamma_1,\Gamma_2$.
The others follow directly 
by induction hypothesis. 
Clearly $(\lambda \x.\M)\N\redbeta \M[\N/\x]$.
By Lemma \ref{lem:generation}.2 $\Pi\leadsto \Pi_1$ ending
as 
$$
\infer{\A_1\der \lambda \x.\M:\sigma\lin A}{
\Pi_2\pder \A_1,\x:\sigma\der \M:A
}
$$
By the Substitution
Lemma \ref{lem:substitution} since 
$\Pi_2\pder \A_1,\x:\sigma\der \M:A$ and 
$\Sigma\pder \A_2\der \N:\sigma$ we have
$
\A_1,\A_2\der \M[\N/\x]:A
$, hence the conclusion follows.
\end{proof}
It is worth noting that, 
due to the additive rule $(\bt E)$, $\BSTA$ is no more correct
for polynomial time, 
since terms with exponential number of reductions 
can be typed by derivations with a priori fixed degree, 
where the degree is the nesting of $(sp)$ applications.
\begin{example}
\label{example}
Consider for $n\in\mathbb{N}$ terms $\M_n$ of the shape: $$(\lambda {\tt f}.\lambda {\tt z}. {\tt f}^{n}{\tt z})(\lambda \x.\ifz \x \thenz \x \elsez\x)\0$$ 
It is easy to verify that for each $\M_n$ there exist reduction sequences
of length exponential in $n$.
\end{example}

%% file: strongNormalization.tex
\subsection{Strong Normalization}
Strong normalization is proved by a translation, preserving reduction, of $\BSTA$
in a slightly variant of Girard's System F \cite{girard72}.
The variant we consider is showed in Fig. \ref{systemF} and it differs 
from the original system since it has explicit rules for weakening and contraction. It is straightforward to prove that it shares all the properties of the original one,
in particular strong normalization. 
\begin{definition}
The types of \emph{System F} are defined by the following grammar:
$$
A,B\ ::=\ \alpha\ |\ A\Rightarrow B \ |\ \forall \alpha.A 
$$
where $\alpha$ ranges over a countable set of type variables.
\end{definition}
\begin{table}
\begin{center}
\begin{tabular}{|c|}
\hline
\\
$
\infer[(Ax)]{{\tt x}:A\derf {\tt x}:A}{}
\qquad 
\infer[(w)]{\A,\x:A\derf {\tt M}:B}{\A \derf {\tt M}:B  }
\qquad
\infer[(c)]{\A, {\tt x}:A\der {\tt M}[{\tt x}/{\tt x}_1,{\tt x}/{\tt x}_2 ]:B}{\A,{\tt x}_1:A,{\tt x}_2:A\der {\tt M}:B}
$
\\
\\
$
\qquad
\infer[(\Rightarrow I)]{\A\derf {\tt \lambda x.M}:A\lin B}{\A, {\tt x}:A\derf {\tt M}:B }
\qquad 
\infer[(\Rightarrow E)]{\A,\B\derf {\tt MN}:B}{\A\derf {\tt M}:A\lin B & \B \der {\tt N}:A }
$
\\
\\
$
\infer[(\forall E)]{\A\der {\tt M}:B[A/\alpha]}
{\A\der {\tt M}:\forall \alpha .B }
\qquad
\infer[(\forall I)]{\A\der {\tt M}:\forall \alpha .A}
{\A\der {\tt M}:A & \alpha\notin \FTV(\A)}

$
\\
\\
\hline
\end{tabular}
\end{center}
\caption{ System F with explicit contraction and weakening rules}
\label{systemF}
\end{table}
We firstly define a forgetful map over types and terms.
\begin{definition}The map $(-)^*$ is defined on types as:
  $$
\begin{array}{c}
(\bt)^*=\forall \alpha. \alpha \Rightarrow\alpha\Rightarrow \alpha
\qquad 
(\alpha)^*=\alpha
\qquad 
(\sigma\lin A)^*= (\sigma)^*\Rightarrow (A)^*
\\[2mm]
(!\sigma)^*=(\sigma)^*
\qquad
(\forall \alpha. A)^*=\forall \alpha.(A)^*
\end{array}
$$
and it is defined on terms as:
$$
\begin{array}{c}
(\0)^*=\lambda \x\y.\x
\qquad 
(\1)^*=\lambda \x\y.\y
\qquad
(\lif{\M}{\M_1}{\M_2})^*= (\M)^*(\M_1)^*(\M_2)^*
\\[2mm]
(\lambda\x.\M)^*=\lambda \x.(\M)^*
\qquad
(\M\N)^*=(\M)^*(\N)^*
\end{array}
$$
\end{definition}
The following lemma assures that the translation well behaves.
\begin{lemma}\label{F-trans}
If $\A \der \M:\sigma$ then $(\A)^*\derf (\M)^*:(\sigma)^*$. 
\end{lemma}
\begin{proof}
  By induction on the derivation $\Pi$ proving $\A \der \M:\sigma$.\\
Let us consider base cases. The $(Ax)$ case is trivial.
Consider the case $\Pi$ consists in the rule
$$
\infer[(\bt_{\0}I)]{\der \0:\bt}{}
$$
Then we have the following derivation
$$
\infer[(\forall I)]{\derf \lambda \x\y.\x:\forall \alpha.\alpha\Rightarrow\alpha\Rightarrow \alpha}{
\infer[(\Rightarrow I)]{\derf \lambda \x\y.\x:\alpha\Rightarrow\alpha\Rightarrow \alpha}{
\infer[(\Rightarrow I)]{\x:\alpha\derf \lambda \y.\x:\alpha\Rightarrow \alpha}{
\infer[(w)]{\y:\alpha,\x:\alpha\derf \x:\alpha
}{
\infer[(Ax)]{\x:\alpha\derf \x:\alpha}{
}
}
}
}
}
$$
The case $\Pi$ consists in the $(\bt_\1I)$ rule is similar. 
The case $\Pi$ ends by $(sp)$ rule follows directly by induction hypothesis. 
The cases where $\Pi$ ends either by a $(\lin I), (\lin E)$ or $(w)$ rules follow by induction hypothesis and an application of the same rule in System F. In the case $\Pi$ ends as
$$
\infer[(\bt E)]{\Gamma \der \lif{{\tt M}}{{\tt N_{\0}}}{{\tt N_{\1}}}:A}{\Gamma\der {\tt M}:\bt & \Gamma\der {\tt N_{\0}}:A  &\Gamma\der {\tt N_{\1}}:A}
$$
we have a derivation ending as:
$$
\infer{(\A)^*\derf (\M)^*(\N_0)^*(\N_1)^*: (A)^*}{
\infer{(\A)^*\derf (\M)^*(\N_0)^*:(A)^*\Rightarrow (A)^*}{
\infer{(\A)^*\derf ({\tt M})^*:(A)^*\Rightarrow (A)^*\Rightarrow (A)^*}{(\A)^*\derf ({\tt M})^*:(\bt)^*=\forall \alpha.\alpha\Rightarrow \alpha\Rightarrow \alpha }
&
(\A)^*\derf (\N_0)^*:(A)^*}
&
(\A)^*\derf (\N_1)^*:(A)^*}
$$
\end{proof}
Moreover, the translation preserves the reduction.
\begin{lemma}[Simulation]\label{simulation}
The following diagrams commutes
$$
\begin{array}{ccc}
\M\ &\ \red\ & \N\ \\[2mm]
\downarrow * & &\downarrow *\\[2mm]
(\M)^{*} & \redbeta^+ & (\N)^*
\end{array}
$$
\end{lemma}
\begin{proof}
The case of a $\beta$-reduction is trivial, so consider a 
$\delta$-reduction as:
$$
\M={\tt R}[\lif{\0}{{\tt P}}{{\tt Q}}]\reddelta {\tt R}[{\tt P}] =\N
$$ 
the other case is analogous.
By definition of the map $(\ )^{*}$ we have:
$$(\M)^{*}=({\tt R}[\lif{\0}{{\tt P}}{{\tt Q}}])^*={\tt R}'[(\0)^{*}({\tt P})^{*}({\tt Q})^{*}]
={\tt R}'[(\lambda \x.\lambda \y.\x)({\tt P})^{*}({\tt Q})^{*}]$$
 and 
clearly:
$$
{\tt R}'[(\lambda \x.\lambda \y.\x)({\tt P})^{*}({\tt Q})^{*}]
\redbeta
{\tt R}'[(\lambda \y.({\tt P})^{*})({\tt Q})^{*}]
\redbeta
{\tt R}'[({\tt P})^{*}]=(\N)^*
$$
and so the conclusion. 
\end{proof}
Now, we have the following.
\begin{theorem}[Strong Normalization]$ $\\
\label{lem:StrongNormalization}
If $\Gamma\vdash \M:\sigma$ then $\M$ is strongly normalizing with respect to the  relation $\red$.
\end{theorem}
\begin{proof}
By Lemmas \ref{F-trans} and \ref{simulation} and the strong normalization of System F.
\end{proof}


%% file: strategy.tex
\section{Structural Operational Semantics}
\label{strategy}
In this section the operational semantics of terms of $\Lambda_{\mathcal{B}}$ is presented,
through an evaluation machine, named $\mathrm{K}_{\mathcal{B}}^{\mathcal{C}}$, defined in SOS style \cite{Plotkin81,Kahn87}.
The machine $\mathrm{K}_{\mathcal{B}}^{\mathcal{C}}$ is related to the type assignment system $\BSTA$ since
it evaluates programs (i.e., closed terms of boolean type). The machine allows us to measure the space used during the evaluation. In order to justify our space measure,
a small step version of $\mathrm{K}_{\mathcal{B}}^{\mathcal{C}}$ is used.
\subsection{The evaluation machine $\mathrm{K}_{\mathcal{B}}^{\mathcal{C}}$}
The  machine $\mathrm{K}_{\mathcal{B}}^{\mathcal{C}}$ evaluates programs according to the leftmost outermost 
strategy.
If restricted to $\lambda$-calculus,
the  machine $\mathrm{K}_{\mathcal{B}}^{\mathcal{C}}$ is quite similar to the Krivine machine \cite{Krivine07}, 
since $\beta$-reduction
is not an elementary step, but the substitution of a term to a variable is performed one occurrence at a time.
The machine $\mathrm{K}_{\mathcal{B}}^{\mathcal{C}}$ uses two memory devices, the m-context and the $\bt$-context, that 
memorize respectively the assignments to variables and the control flow.

\begin{definition} \
\begin{itemize}
\item An \emph{m-context} $\mathcal{A}$ is a sequence of variable assignments of the 
shape $\x:=\M$ where $\M$ is a term and all the variables are distinct. The symbol $\varepsilon$ denotes the empty m-context and
the set of m-contexts is denoted by $\mathrm{Ctx_m}$.\\
The \emph{cardinality} of an m-context $\mathcal{A}$, denoted by $\#(\mathcal{A})$,
is the number of variable assignments in $\mathcal{A}$. The \emph{size} of an m-context $\mathcal{A}$, denoted by $|\mathcal{A}|$, is the 
sum of the size of each variable assignment in $\mathcal{A}$, where
a variable assignment $\x:=\M$ has size $|\M|+1$.
\item Let $\circ$ be a distinguished symbol. The set $\mathrm{Ctx}_{\bt}$ of \emph{$\bt$-contexts} is defined by 
the following grammar:
$$
\mathcal{C}[\circ]::=\circ \ |\ (\lif{\mathcal{C}[\circ]}{{\tt M}}{{\tt N}}){\tt V}_1\cdots {\tt V}_n
$$ 
The size of a $\bt$-context $\mathcal{C}[\circ]$, denoted by $|\mathcal{C}[\circ]|$, 
is the size of the term obtained by replacing the symbol $\circ$ by a variable.\\
The cardinality
of a $\bt$-context $\mathcal{C}[\circ]$, denoted by $\#(\mathcal{C}[\circ])$, is the number of nested $\bt$-contexts in it. 
i.e.: 
$$\#(\circ)=0\qquad \#((\lif{\mathcal{C}[\circ]}{{\tt M}}{{\tt N}}){\tt V}_1\cdots {\tt V}_n)=\#(\mathcal{C}[\circ])+1$$
\end{itemize} 
\end{definition}
It is worth noticing that a $\bt$-contexts  $\mathcal{C}[\circ]$ 
can be seen as a stack of atomic contexts where 
its cardinality $\#(\mathcal{C}[\circ])$ is the height of such a stack.
\
\begin{notation}
The notation $\mathcal{A}_1@\mathcal{A}_2$ is used for the concatenation of the disjoint m-contexts $\mathcal{A}_1$ and 
$\mathcal{A}_2$. Moreover, $[\x:=\M]\in\mathcal{A}$ denotes the fact
that $\x:=\M$ is in the m-context $\mathcal{A}$. 
The notation $\FV(\mathcal{A})$ identifies the set:
$\bigcup_{[\x:=\M]\in\mathcal{A}}\FV(\M)$.\\
As usual, $\mathcal{C}[\M]$ denotes the term obtained by filling the hole $[\circ]$
in $\mathcal{C}[\circ]$ by $\M$. In general we omit the hole $[\circ]$ and we range 
over $\bt$-contexts by $\mathcal{C}$. As expected, $\FV(\mathcal{C})$ denotes the set $\FV(\mathcal{C}[\M])$ for every 
closed term $\M$.
\end{notation}

Note that variable assignments in m-contexts are ordered; this fact allows us to define 
the following closure operation. 
\begin{definition}
Let $\mathcal{A}=\{\x_1:=\N_1,\ldots, \x_n:=\N_n\}$ be an m-context. 
Then, $(-)^{\mathcal{A}}:\Lambda_{\mathcal{B}} \to\Lambda_{\mathcal{B}}$  is 
the map associating
to each term $\M$
the term  $(\M)^{\mathcal{A}}\equiv\M[\N_n/\x_n][\N_{n-1}/\x_{n-1}]\cdots [\N_1/\x_1]$.
\end{definition}
The correct inputs for the machine are programs, defined as follows. 
  \begin{definition}
The set $\prog$ of \emph{programs} is the set of closed terms typable
by the ground type. i.e. $\prog=\{ \M \mid\ \vdash \M:\bt\}$.
\end{definition}

The design of the evaluation machine follows the syntactic shape of programs.

\begin{remark}
It is easy to check that every term has the following shape:
$\lambda \x_{1}...\x_{n}.\zeta {\tt V}_1\cdots{\tt V}_m$, for some $n,m \geq 0$, where $\zeta$ is 
either a boolean $b$, a variable $\x$,  a redex $(\lambda\x.\N){\tt P}$,
or a subterm of the shape $\lif{{\tt P}}{\N_{\0}}{\N_{\1}}$.
It is immediate to check that, if a term is in $\prog$, then $n=0$.
Moreover, if a term in $\prog$ is a normal form, then it coincides with a boolean constant $\tt b$. 
\end{remark}

The evaluation machine
$\mathrm{K}_{\mathcal{B}}^{\mathcal{C}}$ proves statements of the shape:
$$\mathcal{C},\mathcal{A}\models {\M}\Downarrow {\tt b} $$ where $\mathcal{C},\mathcal{A}$
are a $\bt$-context and a m-context respectively, $\M$ is a term, and $\tt b$ is a boolean value.
Its rules are listed in Table \ref{BigStepB}. 
They need some comments, we describes the rules bottom-up. 
The  $(Ax)$ rule is obvious. The $(\beta)$ rule applies when the head of the subject is a $\beta $-redex, then
the association between the bound variable and the argument is remembered in the m-context and 
the body of the term in functional position is evaluated. Note that an $\alpha$-rule is always performed.
The  $(h)$ rule replaces the head occurrence of the head variable by the term associated with it in the m-context.
Rules $(\ifz\0)$ and $(\ifz\1)$ perform the $\delta$ reductions. 
In order to evaluate the test $\M$, a part of the subject is naturally
erased. 
 This  erased information is stored in the $\bt$-context, indeed 
$\bt$-contexts are stacks that permits to store all the branches 
of a computation produced by conditionals.
When the evaluation of the test $\M$ of the current conditional 
is completed, the machine pops the top $\bt$-context and
continues by evaluating the term in the right branch of the computation.
\begin{table*}
\begin{center}
\begin{tabular}{|c|}
\hline 
 \\
 \infer[(Ax)]{\mathcal{C},\mathcal{A}\models {\tt b}\Downarrow {\tt b}}{}
\\
\\
\infer[(\beta)^\S]{\mathcal{C},\mathcal{A}\models (\lambda \x. \M)\N{\tt V}_1\cdots {\tt V}_m\Downarrow {\tt b}}
{
\mathcal{C},\mathcal{A}@\{\x':=\N\}\models \M[\x'/\x] {\tt V}_1\cdots {\tt V}_m\Downarrow {\tt b} }
\\
\\
\infer[(h)]{\mathcal{C},\mathcal{A}\models \x{\tt V}_1\cdots {\tt V}_m\Downarrow {\tt b}}{
\{\x:=\N\}\in \mathcal{A} &
\mathcal{C},\mathcal{A}\models \N {\tt V}_1\cdots {\tt V}_m\Downarrow {\tt b} }
\\
\\
\ \infer[(\ifz\0)]{\mathcal{C},\mathcal{A}\models(\lif{\M}{\N_{\0}}{\N_{\1}}){\tt V}_1\cdots {\tt V}_m \Downarrow \tt b}
{\mathcal{C}[(\lif{[\circ]}{\N_{\0}}{\N_{\1}}){\tt V}_1\cdots {\tt V}_m],\mathcal{A}\models \M \Downarrow \tt \0
& 
\mathcal{C},\mathcal{A}\models \N_{\0}{\tt V}_1\cdots{\tt V}_m \Downarrow \tt b
}
\\
\\
\ \infer[(\ifz\1)]{\mathcal{C},\mathcal{A}\models(\lif{\M}{\N_{\0}}{\N_{\1}}){\tt V}_1\cdots {\tt V}_m \Downarrow \tt b}
{\mathcal{C}[(\lif{[\circ]}{\N_{\0}}{\N_{\1}}){\tt V}_1\cdots {\tt V}_m],\mathcal{A}\models \M \Downarrow \tt \1
& 
\mathcal{C},\mathcal{A}\models \N_{\1}{\tt V}_1\cdots{\tt V}_m \Downarrow \tt b
}

\\
\\
(\S) $\x'$ is a fresh variable.
\\
\hline
\end{tabular}
\end{center}
\caption{The Abstract Machine $\mathrm{K}_{\mathcal{B}}^{\mathcal{C}}$}
\label{BigStepB}
\end{table*}
The behaviour of the machine $\mathrm{K}_{\mathcal{B}}^{\mathcal{C}}$ 
is formalized in the following definition.
\begin{definition}\
\label{MachineDefinition}
\begin{enumerate}
\item The \emph{evaluation relation} $\ev\subseteq \mathrm{Ctx}_{\bt}\times\mathrm{Ctx_m}\times 
\Lambda_{\mathcal{B}} \times \mathcal{B}$  
is the relation inductively defined by the rules of $\mathrm{K}_{\mathcal{B}}^{\mathcal{C}}$. 
If $\M$ is a program, and if there is a boolean $b$ such that $\circ,\varepsilon\models \M\ev b$
then we say that $\M$ \emph{evaluates}, and we write $\M\ev$. 
As usual,
$\models \M\ev b$ is a short for $\circ,\varepsilon\models \M\ev b$.
\item Derivation trees  in the abstract machine are called \emph{computations} and are denoted by $\nabla,\Diamond$.
We use 
$\nabla::\mathcal{C},\mathcal{A}\models \M\ev {\tt b}$ to denote a computation 
with conclusion $\mathcal{C},\mathcal{A}\models \M\ev {\tt b}$.
\item Given a computation $\nabla$  each 
node of $\nabla$, which is of the shape $\mathcal{C},\mathcal{A}\models \M\ev {\tt b}$ is a \emph{configuration}.
The notation $\mathcal{C},\mathcal{A}\models \M\ev {\tt b}\in\nabla$
is used to stress that 
 $\mathcal{C},\mathcal{A}\models \M\ev {\tt b}$ is a configuration in the computation $\nabla$.
Configurations are denoted by $\phi,\psi$. The notation 
$\phi\Yright \mathcal{C},\mathcal{A}\models \M\ev {\tt b}$ means that $\phi$ is
 the configuration
 $\mathcal{C},\mathcal{A}\models \M\ev {\tt b}$.
The
conclusion of the derivation tree is called the \emph{initial configuration}.
\item Given a computation $\nabla$, the \emph{path} to reach a configuration $\phi$
denoted $\pathi_{\nabla}(\phi)$ is 
the sequence of configurations 
between the conclusion of $\nabla$ and $\phi$.
In general,  we simply write $\pathi(\phi)$ when $\nabla$
is clear from the context.
\end{enumerate}
\end{definition}
\input{example}
In Table \ref{TableCompEx}
we present an example of $\mathrm{K}_{\mathcal{B}}^{\mathcal{C}}$ 
computation on a term $\M_2$ as defined in Example \ref{example}.\\

In order to prove that the machine is sound and complete with respect to programs, we
need to prove some additional properties. First of all, the next lemma proves that the machine enjoys 
a sort of weakening, with respect to both contexts.
\begin{lemma}
\begin{enumerate}\label{weakcon}
\item Let $\mathcal{C}[\circ], \mathcal{A}\models  \M\ev \tt b$. Then, for every $\mathcal{C'}[\circ]$ such that
$(\mathcal{C'}[\mathcal{C}[\M]])^{\mathcal{A}} \in \prog$, $\mathcal{C'}[\mathcal{C}[\circ]],\mathcal{A} \models\M \ev\tt b$.
\item  Let $\nabla :: \mathcal{C},\mathcal{A}\models \M\ev {\tt b}$ and
$\x$ be a fresh variable. 
Then, 
 $\nabla :: \mathcal{C},\mathcal{A}@\{\x:=\N \}\models \M\ev {\tt b}$
\end{enumerate}
\end{lemma}
\begin{proof}
Both points can be easily proved by induction on the computation.
\end{proof}

\begin{lemma}\
\label{lem:redProperties}
\begin{enumerate}
\item Let $\mathcal{C},\mathcal{A} \models \M \ev b$ and let $(\mathcal{C}[\M])^{\mathcal{A}}\in \prog$.
Then, both $(\M)^{\mathcal{A}}\red^{*} \tt b$ and $(\mathcal{C}[\M])^{\mathcal{A}}\red^{*} \tt b'$,
for some $\tt b'$.
\item Let $ \M\in\prog$ and $\nabla:: \models \M\ev {\tt b}$. For each 
$\phi\Yright\mathcal{C},\mathcal{A}\models \N\ev {\tt b}'\in \nabla$, 
$(\mathcal{C}[\N])^{\mathcal{A}}\in \prog$.
\item Let $(\M)^{\mathcal{A}} \in \prog$ and 
$(\M)^{\mathcal{A}}\red^{*} {\tt b}$.
Then,
$\circ, \mathcal{A}\models  \M\ev \tt b$.
\end{enumerate}
\end{lemma}
\begin{proof}$ $
\begin{enumerate}
\item 
First of all, the property $(\mathcal{C}[\M])^{\mathcal{A}}\red^{*}\tt b'$, for some $\tt b'$ derives directly from
the fact that 
$(\mathcal{C}[\M])^{\mathcal{A}}\in \prog$.
In fact this implies $(\mathcal{C}[\M])^{\mathcal{A}}$ is a closed strongly normalizing term of type $\bt$, 
and so its normal form is necessarily a boolean constant. 
So in what follows we will prove just that $\mathcal{C},\mathcal{A} \models \M \ev \tt b$ and $(\mathcal{C}[\M])^{\mathcal{A}}\in \prog$
implies $(\M)^{\mathcal{A}}\red^{*} \tt b$.
Note that if $(\mathcal{C}[\M])^{\mathcal{A}}\in \prog$ then clearly 
$(\M)^{\mathcal{A}}\in \prog$.
We proceed by induction on the derivation proving  $\mathcal{C},\mathcal{A} \models \M \ev \tt b$. 
Let the last rule be:  $$\infer[(Ax)]{\mathcal{C},\mathcal{A}\models {\tt b}\Downarrow {\tt b}}{}$$ 
Obviously $(\tt b)^{\mathcal{A}}\red^{*}\tt b$. 
Let the derivation ends as:
$$
\infer[(\beta)]{\mathcal{C},\mathcal{A}\models (\lambda \x. {\tt P})\N{\tt V}_1\cdots {\tt V}_m\Downarrow {\tt b}}
{
\mathcal{C},\mathcal{A}@[\x':=\N]\models {\tt P}[\x'/\x] {\tt V}_1\cdots {\tt V}_m\Downarrow {\tt b}}
$$
By induction hypothesis
$({\tt P}[\x'/\x] {\tt V}_1\cdots {\tt V}_m)^{\mathcal{A}@[\x':=\N]}\red^{*} {\tt b}$.
Clearly since $\x'$ is fresh:
$$({\tt P}[\x'/\x]{\tt V}_1\cdots {\tt V}_m )^{\mathcal{A}@[\x':=\N ]}
\equiv
(({\tt P}[\x'/\x]{\tt V}_1\cdots {\tt V}_m )[\N/\x'])^{\mathcal{A}}
\equiv
({\tt P}[\N/\x]{\tt V}_1\cdots {\tt V}_m )^{\mathcal{A}}
$$
hence:
$$((\lambda \x. {\tt P})\N{\tt V}_1\cdots {\tt V}_m )^{\mathcal{A}}\red ({\tt P}[\N/\x]{\tt V}_1\cdots {\tt V}_m )^{\mathcal{A}}\red^{*} {\tt b}$$
and the conclusion follows. 
The case of a rule  $(h)$ follows directly by induction hypothesis.\\
Let the derivation end as:
$$
\infer[(\ifz\ \0)]{\mathcal{C},\mathcal{A}\models(\lif{{\tt P}}{\N_{\0}}{\N_{\1}}){\tt V}_1\cdots {\tt V}_m \Downarrow \tt b}
{\mathcal{C}',\mathcal{A}\models {\tt P} \Downarrow \tt \0
& 
\mathcal{C},\mathcal{A}\models \N_{\0}{\tt V}_1\cdots{\tt V}_m \Downarrow \tt b
}
$$
where $\mathcal{C}'=\mathcal{C}[(\lif{[\circ]}{\N_{\0}}{\N_{\1}}){\tt V}_1\cdots {\tt V}_m]$.
By induction hypothesis
$({\tt P})^{\mathcal{A}}\red^{*}\0$,
 hence: 
$$((\lif{{\tt P}}{\N_{\0}}{\N_{\1}}){\tt V}_1\cdots {\tt V}_m)^{\mathcal{A}}\red^{*}
((\lif{\0}{\N_{\0}}{\N_{\1}}){\tt V}_1\cdots {\tt V}_m)^{\mathcal{A}}
$$ and by $\delta$ reduction
$$((\lif{\0}{\N_{\0}}{\N_{\1}}){\tt V}_1\cdots {\tt V}_m)^{\mathcal{A}} \reddelta
(\N_{\0}{\tt V}_1\cdots {\tt V}_m)^{\mathcal{A}}
$$
 moreover, since by induction hypothesis we also have 
 $(\N_{\0}{\tt V}_1\cdots{\tt V}_m)^{\mathcal{A}} \red^{*} {\tt b}$, the conclusion follows. 
The case of rule $(\ifz\ \1)$ is analogous.
\item Easy, by induction on the length of $\pathi(\phi)$.
\item The proof is by induction on the number of steps needed to reach
the normal form ${\tt b}$  of $(\M)^{\mathcal{A}}$ 
according to the leftmost strategy. Since $(\M)^{\mathcal{A}}$
is strongly normalizing this is clearly well-founded.\\
If $(\M)^{\mathcal{A}}$ is already in normal form, since it is must be typable of type $\bt$ then $\M \equiv \tt b$, and the result is trivial.
Otherwise $(\M)^{\mathcal{A}}$ cannot be an abstraction, since its typing, so it is an application $\N \Q {\tt V}_{1}...{\tt V}_{m}$.\\
Suppose
$\N \equiv \lambda \x.\R$. There are two cases, either 
$(\M)^{\mathcal{A}}\equiv ((\lambda \x.\R')\Q' {\tt V}'_{1}...{\tt V}'_{m})^{\mathcal{A}}$
or $(\M)^{\mathcal{A}}\equiv (\y\Q' \V'_{1}...\V'_{m})^{\mathcal{A}}$ and $\{\y:=\lambda \x.\R'\}\in \mathcal{A}$.\\
Let us consider the first case. Then $((\lambda \x.\R')\Q' {\tt V}'_{1}...{\tt V}'_{m})^{\mathcal{A}} \redbeta 
(\R'[\Q'/\x] \V'_{1}...\V'_{m})^{\mathcal{A}} \equiv
(\R'[\x'/\x] \V'_{1}...\V'_{m})^{\mathcal{A}@\{\x':=\Q' \}}$. By induction hypothesis we have
$[\circ], \mathcal{A}@\{\x':=\Q' \}\models \R'[\x'/\x] \V'_{1}...\V'_{m}  \ev \tt b$
and so the result follows by rule $(\beta)$.\\
In the second case, since $\{\y:=\lambda \x.\R'\}\in \mathcal{A}$, then $(\y\Q' \V'_{1}...\V'_{m})^{\mathcal{A}} \redbeta 
(\R'[\Q'/\x] \V'_{1}...\V'_{m})^{\mathcal{A}}\equiv (\R'[\x'/\x] \V'_{1}...\V'_{m})^{\mathcal{A}@\{\x':=\Q' \}} $.
By induction hypothesis
$[\circ], \mathcal{A}@\{\x':=\Q' \}\models \R'[\x'/\x] \V'_{1}...\V'_{m}\ev \tt b$, so by one application of the rule 
$(\beta)$, $[\circ], \mathcal{A}\models (\lambda \x.\R' ){\tt Q'}\V'_{1}...\V'_{m}\ev \tt b$. Finally, by one application of the rule $(h)$,
since $\{\y:=\lambda \x.\R'\}\in \mathcal{A}$, we have
 $[\circ], \mathcal{A}\models \y\Q' \V'_{1}...\V'_{m}\ev \tt b$.\\
The remaining case is the one where $\N \equiv \lif{\M'}{\N'_{\0}}{\N'_{\1}}$. 
By definition  
$(\M)^{\mathcal{A}}\equiv 
((\lif{\M'}{\N'_{\0}}{\N'_{\1}})\Q' \V'_{1}...\V'_{m})^{\mathcal{A}}\equiv
(\lif{(\M')^{\mathcal{A}}}{(\N'_{\0})^{\mathcal{A}}}{(\N'_{\1})^{\mathcal{A}}})(\Q')^{\mathcal{A}}
 (\V'_{1})^{\mathcal{A}}...(\V'_{m})^{\mathcal{A}}$. Since $(\M)^{\mathcal{A}}\in \prog$, 
$\der (\M)^{\mathcal{A}}:\bt$, so, by the strong normalization property, $(\M)^{\mathcal{A}}\red^{*}\tt b $.
This implies either $(\M')^{\mathcal{A}}=\tt b' $ or $(\M')^{\mathcal{A}}\red^{*} \tt b' $ for some $\tt b'$. Let us consider the latter 
case. The number of reduction steps of the sequence $(\M')^{\mathcal{A}}\red^{*} \tt b' $ is shorter than that one of
$(\M)^{\mathcal{A}}\red^{*} \tt b $, so by induction $[\circ], \mathcal{A} \models \M' \ev \tt b'$, and, by  Lemma \ref{weakcon}.1, $(\lif{[\circ]}{\N'_{\0}}{\N'_{\1}})\Q' \V'_{1}...\V'_{m} , \mathcal{A} \models \M' \ev \tt b'$.
Without loss of generality, we consider only the case where $\tt b' = \0$. Then $(\M)^{\mathcal{A}}\red^{*}\tt b $ implies $(\N'_{\0}\Q' {\tt V}'_{1}...{\tt V}'_{m})^{\mathcal{A}}
\red^{*} \tt b $, so by induction $[\circ], \mathcal{A} \models \N'_{\0}\Q' {\tt V}'_{1}...{\tt V}'_{m} \ev \tt b$, and the result 
follows by rule $(\ifz\0)$. The case $(\M')^{\mathcal{A}}=\tt b' $ is easier.
The case $(\M)^{\mathcal{A}}\equiv 
(\y\Q' \V'_{1}...\V'_{m})^{\mathcal{A}}$, and $(\y:= \lif{\M'}{\N'_{\0}}{\N'_{\1}})$, is similar, but both rules $(h)$ and $(\ifz )$
must be used.
\qed
\end{enumerate}
\end{proof}

Then we can state the soundness and completeness of the evaluation 
machine $\mathrm{K}_{\mathcal{B}}^{\mathcal{C}}$ with respect to the reduction on programs.
\begin{theorem}Let $\M\in \prog\ $. Then:
\begin{enumerate}
\item If $\models\M\ev{\tt b}$ then
$\M\red^* {\tt b}$.  \hspace{6cm} {\rm {(Soundness)}} 
\item If $\M\red^* {\tt b}$ then
$\models\M\ev{\tt b}$. \hspace{5.5cm} {\rm {(Completeness)}}
\end{enumerate}
\end{theorem}
\begin{proof}$ $
\begin{enumerate}
\item  It follows directly by Lemma \ref{lem:redProperties}.(1).
\item  It follows directly by Lemma \ref{lem:redProperties}.(3). \qed
\end{enumerate}
\end{proof}

%% file: example.tex
\begin{table*}
\begin{center}
\begin{tabular}{|c|}
\hline \\
\infer{\models (\lambda {\tt f}.\lambda {\tt z}. {\tt f}^2{\tt z})(\lambda \x.\ifz \x \thenz \x \elsez\x)\0 \ev \0}{
  \infer{\mathcal{A}_0\models (\lambda z.{\tt f}_1({\tt f}_1{\tt z}))\0 \ev \0}{
  \infer{\mathcal{A}_1\models {\tt f}_1({\tt f}_1{\tt z}_1) \ev \0}{
    \infer{\mathcal{A}_1 \models (\lambda \x.\ifz \x \thenz \x \elsez\x)({\tt f}_1{\tt z}_1)\ev \0}{
      \infer{\mathcal{A}_2\models \ifz \x_1 \thenz \x_1 \elsez\x_1 \ev \0}{
\infer{\mathcal{C}_0,\mathcal{A}_2\models \x_1 \ev \0  }{
\infer{\mathcal{C}_0,\mathcal{A}_2\models {\tt f}_1{\tt z}_1\ev \0 }{
\infer{\mathcal{C}_0,\mathcal{A}_2\models (\lambda \x.\ifz \x \thenz \x \elsez \x){\tt z}_1\ev \0 }{
  \infer{\mathcal{C}_0,\mathcal{A}_3\models \ifz \x_2 \thenz \x_2 \elsez \x_2\ev \0 }{
\infer{\mathcal{C}_1,\mathcal{A}_3\models \x_2 \ev \0 }{
\infer{\mathcal{C}_1,\mathcal{A}_3\models {\tt z}_1 \ev \0 }{
\infer{\mathcal{C}_1,\mathcal{A}_3\models \0\ev \0 }{
}
}
}
&
\infer{\mathcal{C}_0,\mathcal{A}_3\models \x_2 \ev \0 }{
\infer{\mathcal{C}_0,\mathcal{A}_3\models {\tt z}_1 \ev \0 }{
  \infer{\phi\rhd \mathcal{C}_0,\mathcal{A}_3\models \0\ev \0 }{
}}}
}
}
}
}
&
\infer{\psi\rhd\mathcal{A}_2\models \x_1 \ev \0  }{
\infer{\mathcal{A}_2\models {\tt f}_1{\tt z}_1\ev \0 }{
\infer{\mathcal{A}_2\models (\lambda \x.\ifz \x \thenz \x \elsez \x){\tt z}_1\ev \0 }{
\infer{\mathcal{A}_4\models \ifz \x_3 \thenz \x_3 \elsez \x_3\ev \0 }{
\infer{\mathcal{C}_2,\mathcal{A}_4\models \x_3 \ev \0 }{
\infer{\mathcal{C}_2,\mathcal{A}_4\models {\tt z}_1 \ev \0 }{
\infer{\mathcal{C}_2,\mathcal{A}_4\models \0\ev \0 }{
}
}
}
&
\infer{\mathcal{A}_4\models \x_3 \ev \0 }{
\infer{\mathcal{A}_4\models {\tt z}_1 \ev \0 }{
\infer{\mathcal{A}_4\models \0\ev \0 }{
}}}
}
}
}
}
}
}
}
}
}\\
\\
\hline\\

\begin{tabular}{c}
\begin{tabular}{l}
$\mathcal{A}_0=[{\tt f}_1:=\lambda \x.\ifz \x \thenz \x \elsez\x ]$\\
$\mathcal{A}_1=\mathcal{A}_0@[{\tt z}_1:=\0]$\\
$\mathcal{A}_2=\mathcal{A}_1@[\x_1:={\tt f}_1{\tt z}_1]$\\
$\mathcal{A}_3=\mathcal{A}_2@[\x_2:={\tt z}_1]$\\
$\mathcal{A}_4=\mathcal{A}_2@[\x_3:={\tt z}_1]$
\end{tabular}
\begin{tabular}{l}
$\mathcal{C}_0=\lif{\circ}{\x_1}{\x_1}$\\
$\mathcal{C}_1=\mathcal{C}_0[\lif{\circ}{\x_2}{\x_2}]$\\
$\mathcal{C}_2=\lif{\circ}{\x_3}{\x_3}$\\
$ $\\
$ $
\end{tabular}
\end{tabular}\\
\\
\hline
\end{tabular}
\end{center}
\caption{An example of computation in $\mathrm{K}_{\mathcal{B}}^{\mathcal{C}}$.}
\label{TableCompEx}
\end{table*}

%% file: smallstep.tex
\subsection{A small step version of $\mathrm{K}_{\mathcal{B}}^{\mathcal{C}}$}
The proof that programs are evaluated by the machine $\mathrm{K}_{\mathcal{B}}^{\mathcal{C}}$ in polynomial space needs a formal definition of the space
consumption, which in its turn needs a deep investigation on the 
machine behaviour. In fact, we will explicitly show 
that computations in the machine
$\mathrm{K}_{\mathcal{B}}^{\mathcal{C}}$ can be performed 
with no need of backtracking or complex state memorizations.\\
For this reason, in Table \ref{SmallStepMachine} we depict a small step abstract
machine $\mathrm{k}_{\mathcal{B}}^{\mathcal{C}}$ that is able to reduce 
sequentially programs in $\STA_\bt$ following a leftmost outermost strategy and that 
exploit  a use of contexts similar to the one 
implemented by the machine $\mathrm{K}_{\mathcal{B}}^{\mathcal{C}}$.
The rules are similar to the ones in Table  \ref{BigStepB}
but we need a further stack 
in order to maintain the desired complexity property.\\
In what follows we show that every big step computation has its small step correspondent.
So, the small step machine by making explicit
the evaluation order
clarifies the fact that  
every configuration 
depends uniquely on the previous one (thanks to the use of contexts).
From this we can deduce that the space needed in order 
to evaluate a program is the maximum space used by one of its 
configurations.\\
The big step machine has the advantage of being 
more abstract and this makes it easy to prove the complexity properties.
In fact, the use of a further stack  makes more difficult
the proofs of such properties for the small step machine.
For this reason in what follows we
prefer to work on the big step machine. 
\begin{table*}
\begin{center}
\small
\begin{tabular}{|c|}
\hline 
\\
\infer[(\beta)^{\S}]{\langle \mathcal{S},\mathcal{C},\mathcal{A}\succ (\lambda \x. \M)\N{\tt V}_1\cdots {\tt V}_m\rangle
\mapsto 
\langle \mathcal{S},\mathcal{C},\mathcal{A}@[\x':=\N]\succ \M[\x'/\x] {\tt V}_1\cdots {\tt V}_m\rangle}{
}
\\
\\
\infer[(h)]{\langle \mathcal{S},\mathcal{C},\mathcal{A}\succ \x{\tt V}_1\cdots {\tt V}_m\rangle\mapsto \langle  \mathcal{S}, \mathcal{C},\mathcal{A}\succ \N {\tt V}_1\cdots {\tt V}_m\rangle }{
[\x:=\N]\in\overline{\mathcal{S}\cdot \mathcal{A}}
}
\\
\\
\infer[(\ifz)]{\langle \mathcal{S},\mathcal{C},\mathcal{A}\succ(\lif{\M}{\N_{\0}}{\N_{\1}}){\tt V}_1\cdots {\tt V}_m \rangle \mapsto \langle\mathcal{S}\cdot\mathcal{A}, \mathcal{C}',\epsilon\succ \M\rangle}
{
\mathcal{C}'=\mathcal{C}[(\lif{[\circ]}{\N_\0}{\N_\1}){\tt V}_1\cdots{\tt V}_n]
}
\\
\\
\infer[(r_\0)]{ \langle  \mathcal{S}\cdot  \mathcal{A}, \mathcal{C}[(\lif{[\circ]}{\N_\0}{\N_\1}){\tt V}_1\cdots{\tt V}_n],\mathcal{A}'\succ \0\rangle
\mapsto
 \langle  \mathcal{S},\mathcal{C},\mathcal{A}\succ \N_\0{\tt V}_1\cdots {\tt V}_n\rangle}{
}
\\
\\
\infer[(r_\1)]{ \langle  \mathcal{S}\cdot  \mathcal{A}, \mathcal{C}[(\lif{[\circ]}{\N_\0}{\N_\1}){\tt V}_1\cdots{\tt V}_n ],\mathcal{A}'\succ \1\rangle
\mapsto
 \langle \mathcal{S},\mathcal{C},\mathcal{A} \succ \N_\1{\tt V}_1\cdots {\tt V}_n\rangle}{
}
\\
\\
(\S) $\x'$ is a fresh variable.\\
\hline
\end{tabular}
\normalsize
\end{center}
\caption{The small step machine $\mathrm{k}_{\mathcal{B}}^{\mathcal{C}}$}
\label{SmallStepMachine}
\end{table*}
In order to state formally the behaviour of the machine $\mathrm{k}_{\mathcal{B}}^{\mathcal{C}}$ we need to define a further stack containing m-contexts,
this is done in the following
definition.
\begin{definition}\
\begin{itemize}
\item An \emph{m-stack} $\mathcal{S}$ is a stack of m-contexts.
The symbol $\epsilon$ denote the empty m-stack. The expression
$\mathcal{S}\cdot \mathcal{A}$ denotes the operation of 
pushing the m-context $\mathcal{A}$ on the m-stack $\mathcal{S}$.
The set of m-stacks is denoted by $\mathrm{Stk_m}$.
The expression $\overline{\mathcal{S}}$ denotes 
the m-context obtained by concatenating all the m-context
in $\mathcal{S}$, i.e. $\overline{\epsilon}=\varepsilon$ and
$\overline{\mathcal{S}\cdot \mathcal{A}}= \overline{\mathcal{S}}@\mathcal{A}$.  
\item
 The \emph{reduction relation} $\mapsto\subseteq (\mathrm{Stk_m}\times\mathrm{Ctx}_{\bt}\times\mathrm{Ctx_m}\times 
\Lambda_{\mathcal{B}})\times (\mathrm{Stk_m}\times\mathrm{Ctx}_{\bt}\times\mathrm{Ctx_m}\times 
\Lambda_{\mathcal{B}})$  
is the relation inductively defined by the rules of $\mathrm{k}_{\mathcal{B}}^{\mathcal{C}}$. The relation $\mapsto^*$ is the reflexive and transitive closure of the reduction relation $\mapsto$.\\
If $\M$ is a program, and if there is a boolean $b$ such that 
$\langle \epsilon, \circ,\varepsilon\succ \M\rangle \mapsto^*\langle \epsilon, \circ,\mathcal{A}\succ b\rangle $ for some $\mathcal{A}$, 
then we say that $\M$ \emph{reduces} to $b$, and we simply write $\M\mapsto^* b$ for short.
\end{itemize}
\end{definition}
We now prove that we have a direct correspondence between 
the configurations of a computation in the 
big step machine and the small step machine configurations.\\
Given a big step abstract machine derivation 
$\nabla :: \models \M\ev {\tt b}$, 
we can define a translation $(-)^{\tt s}$ assigning 
to each
configuration 
$
\phi\Yright \mathcal{C},\mathcal{A}\models {\tt N}\ev {\tt b'}\in \nabla
$ 
a 
small step abstract  machine configuration 
$\langle \mathcal{S},\mathcal{C},\mathcal{A}'\succ \N\rangle$
such that $\overline{\mathcal{S}\cdot\mathcal{A}'}=\mathcal{A}$.
Let $(-)^{\tt s}$ be 
inductively defined, for every configuration $\phi\in\nabla$, on 
the length $n$ of ${\tt path}_\nabla (\phi)$ as:
\begin{itemize}
\item if $n=1$ then
$$(\circ,\varepsilon \models \M\ev {\tt b})^{\tt s}=\langle \epsilon,\circ,\varepsilon\succ \M\rangle $$
\item if $n=m+1$ then for some $\psi$  we have 
${\tt path}_{\nabla}(\phi)={\tt path}_{\nabla}(\psi)+1$ and in particular
we have a rule $(R)$ like the following 
$$
\infer[(R)]{\psi\Yright \mathcal{C},\mathcal{A}\models \N\Downarrow {\tt b}}{ \mathcal{C}_1,\mathcal{A}_1\models \N_1\Downarrow {\tt b}_1
&\cdots &
\mathcal{C}_k,\mathcal{A}_k\models \N_k\Downarrow {\tt b}_k
}
$$
for $1\leq k\leq 2$ where the length of 
${\tt path}_{\nabla}(\psi)$ is $m$ and 
$\phi$ is one of the premise configurations of $(R)$.
We now proceed by case on $(R)$.\\
If $(R)$ is the rule:
$$\infer[(\beta)]{\psi\Yright \mathcal{C},\mathcal{A}\models (\lambda \x. \M)\N{\tt V}_1\cdots {\tt V}_m\Downarrow {\tt b}}
{
\phi\Yright\mathcal{C},\mathcal{A}@\{\x':=\N\}\models \M[\x'/\x] {\tt V}_1\cdots {\tt V}_m\Downarrow {\tt b} }
$$
Then, by induction hypothesis we have $(\psi)^{\tt s}= \langle \mathcal{S},\mathcal{C},\mathcal{A}'\succ (\lambda \x.\M)\N{\tt V}_1\cdots {\tt V}_m\rangle$ such that $\overline{\mathcal{S}\cdot \mathcal{A}'}=\mathcal{A}$, so we can define 
$$(\phi)^{\tt s}=\langle \mathcal{S},\mathcal{C},\mathcal{A}'@[\x':=\N]\succ \M[\x'/\x]{\tt V}_1\cdots {\tt V}_m\rangle $$
and clearly $\overline{\mathcal{S}\cdot (\mathcal{A}'@\{\x':=\N\})}=\mathcal{A}\{\x':=\N\}$.\\
If $(R)$ is the rule:
$$
\infer[(h)]{\psi\Yright \mathcal{C},\mathcal{A}\models \x{\tt V}_1\cdots {\tt V}_m\Downarrow {\tt b}}{
\{\x:=\N\}\in \mathcal{A} &
\phi\Yright \mathcal{C},\mathcal{A}\models \N {\tt V}_1\cdots {\tt V}_m\Downarrow {\tt b} }
$$
Then,  by induction hypothesis we have $(\psi)^{\tt s}= \langle \mathcal{S},\mathcal{C},\mathcal{A}'\succ \x{\tt V}_1\cdots {\tt V}_m\rangle$ such that $\overline{\mathcal{S}\cdot \mathcal{A}'}=\mathcal{A}$
so we can define 
$$(\phi)^{\tt s}=\langle \mathcal{S},\mathcal{C},\mathcal{A}'\succ \N{\tt V}_1\cdots {\tt V}_m\rangle $$
If $(R)$ is the rule:
$$
\ \infer[(\ifz\0)]{\psi\Yright \mathcal{C},\mathcal{A}\models(\lif{\M}{\N_{\0}}{\N_{\1}}){\tt V}_1\cdots {\tt V}_m \Downarrow \tt b}
{\mathcal{C}[(\lif{[\circ]}{\N_{\0}}{\N_{\1}}){\tt V}_1\cdots {\tt V}_m],\mathcal{A}\models \M \Downarrow \tt \0
& 
\mathcal{C},\mathcal{A}\models \N_{\0}{\tt V}_1\cdots{\tt V}_m \Downarrow \tt b
}
$$
by induction hypothesis $(\psi)^{\tt s}= \langle \mathcal{S},\mathcal{C},\mathcal{A}'\succ (\lif{\M}{\N_{\0}}{\N_{\1}}){\tt V}_1\cdots {\tt V}_m\rangle$ such that 
$\overline{\mathcal{S}\cdot \mathcal{A}'}=\mathcal{A}$
and we have two distinct cases. 
Consider the case that $\phi \Yright \mathcal{C}[(\lif{[\circ]}{\N_{\0}}{\N_{\1}}){\tt V}_1\cdots {\tt V}_m],\mathcal{A}\models \M \Downarrow \tt \0$, then we can define
$$(\phi)^{\tt s}=\langle \mathcal{S}\cdot\mathcal{A}',\mathcal{C}[(\lif{[\circ]}{\N_{\0}}{\N_{\1}}){\tt V}_1\cdots {\tt V}_m],\varepsilon \succ \M\rangle $$
Analogously, if 
$\phi \Yright \mathcal{C},\mathcal{A}\models \N_{\0}{\tt V}_1\cdots{\tt V}_m \Downarrow \tt b$ then we can define
$$(\phi)^{\tt s}=\langle \mathcal{S},\mathcal{C},\mathcal{A}' \succ \N_{\0}{\tt V}_1\cdots{\tt V}_m\rangle $$
The case $(R)$ is the rule $(\ifz\1)$ is similar.
\end{itemize}
The translation defined above is useful in order to state the correspondence between the big step and the small step machine.
In order to establish this correspondence we need to visit the 
evaluation trees of the big step machine computation following a determined visiting order. In particular, we consider the left-depth-first visit. E.g. consider the following tree:
$$\scalebox{0.7}{
\xymatrix{
 & g\ar@{-}[dr] &&h\ar@{-}[dl]  \\
e\ar@{-}[dr] &  &f\ar@{-}[dl] & \\
& d\ar@{-}[d] & & l\ar@{-}[d] \\
& c\ar@{-}[dr] & & i\ar@{-}[dl]\\
 & & b\ar@{-}[d] &   \\
 & &  a           &   }
}
$$
the left-depth-first visit coincides with the visit of the nodes in the alphabetical order.
Below, we need to talk about the visit of nodes in a given computation  $\nabla\mathop{::} \models \M\ev {\tt b}$. For this reason, we say that a configuration $\psi$ \emph{immediately follows} a configuration $\phi$ if the node visited after $\phi$ for left-depth-first visit is the node $\psi$. For instance, the node $i$ immediatly follows the node $h$ in the above figure.\\
Now we can state an important result.
\begin{lemma}
\label{lem:left-depth-first}
Let $\nabla\mathop{::} \models \M\ev {\tt b}$
and let $\phi,\psi\in \nabla$ be two distinct configurations (i.e. $\phi\neq\psi$) 
such that $\psi$ immediately follows $\phi$ in the left-depth-first visit of $\nabla$.
Then:
$$
(\phi)^{\tt s}\mapsto (\psi)^{\tt s}
$$
\end{lemma}
\begin{proof}
We proceed by induction on the 
height of $\nabla$.
The base case is easy, since 
$\nabla$ is an application of the $(Ax)$ rule, 
hence there are no configurations $\phi,\psi\in \nabla$ such that $\phi\neq\psi$.
Consider now the case where the height of $\nabla$ is greater than 1.
If the rule with conclusion $\phi$ is not an axiom, then $\psi$
coincides with one of its premises. Let us consider all possible cases.
Consider the case where the rule with conclusion $\phi$ is $(\beta)$. Then, we are in a situation as:
$$\infer[(\beta)]{\phi\Yright \mathcal{C},\mathcal{A} \models (\lambda \x. {\tt P})\N{\tt V}_1\cdots {\tt V}_m\Downarrow {\tt b}}
{
\psi\Yright \mathcal{C},\mathcal{A}@\{\x':=\N\}\models {\tt P}[\x'/\x] {\tt V}_1\cdots {\tt V}_m\Downarrow {\tt b} }
$$
then 
\begin{multline*}
(\phi)^{\tt s}=\langle \mathcal{S},\mathcal{C},\mathcal{A}' \succ (\lambda\x. {\tt P})\N{\tt V}_1\cdots {\tt V}_m\rangle
\mapsto\\
\langle
 \mathcal{S},\mathcal{C},\mathcal{A}'@\{\x':=\N\}\succ {\tt P}[\x'/\x] {\tt V}_1\cdots {\tt V}_m\rangle 
=(\psi)^{\tt s}
\end{multline*}
where $\mathcal{A}=\overline{\mathcal{S}\cdot\mathcal{A}'}$.
Consider the case where the rule with conclusion $\phi$ is:
$$
\infer[(h)]{\phi\Yright \mathcal{C},\mathcal{A}\models \x{\tt V}_1\cdots {\tt V}_m\Downarrow {\tt b}}{
\{\x:=\N\}\in \mathcal{A} &
\psi\Yright \mathcal{C},\mathcal{A}\models \N {\tt V}_1\cdots {\tt V}_m\Downarrow {\tt b} }
$$
then 
$$
(\phi)^{\tt s}=\langle \mathcal{S},\mathcal{C},\mathcal{A}' \succ \x{\tt V}_1\cdots {\tt V}_m\rangle
\mapsto
\langle
 \mathcal{S},\mathcal{C},\mathcal{A}'\succ {\tt N}{\tt V}_1\cdots {\tt V}_m\rangle 
=(\psi)^{\tt s}
$$
thanks to the fact that $\{\x:=\N\}\in \mathcal{A}=\overline{\mathcal{S}\cdot\mathcal{A}'}$.\\
If  the rule with conclusion $\phi$ is:
$$
\ \infer[(\ifz\0)]{\phi\Yright \mathcal{C},\mathcal{A} \models(\lif{\N}{\N_{\0}}{\N_{\1}}){\tt V}_1\cdots {\tt V}_m \Downarrow \tt b}
{\psi\Yright \mathcal{C}[(\lif{[\circ]}{\N_{\0}}{\N_{\1}}){\tt V}_1\cdots {\tt V}_m],\mathcal{A}\models \N \Downarrow \tt \0
& 
\mathcal{C},\mathcal{A}\models \N_{\0}{\tt V}_1\cdots{\tt V}_m \Downarrow \tt b
}
$$
then 
\begin{multline*}
(\phi)^{\tt s}=\langle \mathcal{S}, \mathcal{C},\mathcal{A}' \succ (\lif{\N}{\N_{\0}}{\N_{\1}}){\tt V}_1\cdots {\tt V}_m\rangle
\\
\mapsto
\langle
\mathcal{S}, \mathcal{C}[\lif{[\circ]}{\N_{\0}}{\N_{\1}}){\tt V}_1\cdots {\tt V}_m],\mathcal{A}'\succ \N
\rangle 
=(\psi)^{\tt s}
\end{multline*}
The case  of the rule $(\ifz\1)$ is analogous.\\
Now consider the case $\phi$ is the conclusion of an axiom rule, i.e.:
$$
\infer[(Ax)]{\phi\Yright \mathcal{C},\mathcal{A}\models {\tt b}\Downarrow {\tt b}}{}
$$ 
If $\mathcal{C}$ is empty, then $\phi$ is the last configuration
in the left-depth-first visit of $\nabla$, hence 
there is no configuration $\psi\in \nabla$ following $\phi$ such that $\phi\neq\psi$.
Otherwise, 
$\nabla$ has a subderivation $\Diamond$ of the shape: 
$$
\infer[(\ifz{\tt b})]{\phi'\Yright\mathcal{C}',\mathcal{A}' \models(\lif{\N}{\N_{\0}}{\N_{\1}}){\tt V}_1\cdots {\tt V}_m \Downarrow {\tt b}'}
{
\infer{\qquad \vdots\qquad}{ \infer[(Ax)]{\phi\Yright \mathcal{C},\mathcal{A}\models {\tt b}\Downarrow {\tt b}}{}}
& 
\psi\Yright \mathcal{C}',\mathcal{A}'\models \N_{\tt b}{\tt V}_1\cdots{\tt V}_m \Downarrow \tt b'
}
$$
where by definition of left-depth-first visit $\psi$ is the configuration 
following $\phi$. 
Note in particular that 
${\tt path}_{\Diamond}(\phi)$ does not cross 
any $\tt if$-rule by following its left premise. 
In particular, by definition of the translation $(-)^{\tt s}$ this means that if 
$(\phi')^{\tt s}=
\langle \mathcal{S}, \mathcal{C},\mathcal{A}' \succ
(\lif{\N}{\N_{\0}}{\N_{\1}}){\tt V}_1\cdots {\tt V}_m\rangle$
then 
$(\phi)^{\tt s}=
\langle \mathcal{S}\cdot \mathcal{A}', \mathcal{C}[(\lif{[\circ]}{\N_{\0}}{\N_{\1}}){\tt V}_1\cdots {\tt V}_m],\mathcal{A}'' \succ {\tt b}
\rangle$ for some $\mathcal{A}''$
so in particular we have
\begin{multline*}
(\phi)^{\tt s}=\langle  \mathcal{S}\cdot \mathcal{A}', \mathcal{C}[(\lif{[\circ]}{\N_{\0}}{\N_{\1}}){\tt V}_1\cdots {\tt V}_m],\mathcal{A}'' \succ {\tt b} \rangle
\\
\mapsto
\langle
\mathcal{S}, \mathcal{C},\mathcal{A}'\succ \N_{\tt b}{\tt V}_1\cdots {\tt V}_m
\rangle 
=(\psi)^{\tt s}
\end{multline*}
and the proof is given.
\end{proof}

We can now use this result to prove that computations
in the big step machine correspond to computations in the small step one.
\begin{theorem}Let $\M\in\prog$. Then:
$$\models \M\ev{\tt b} \textrm{ implies }  \M \mapsto^{*} {\tt b}$$
\end{theorem}
\begin{proof} 
By repeatedly applying Lemma \ref{lem:left-depth-first}. 
\end{proof}
A converse of the above lemma can be easily obtained. Nevertheless,
the previous result is sufficient in order to show that our space
measures are sound.
Indeed, Lemma \ref{lem:left-depth-first}, if repeatedly applied, allows us to define 
the execution in the $\mathrm{K}_{\mathcal{B}}^{\mathcal{C}}$ machine as a sequence of configurations
corresponding to the left-depth-first visit of the derivation tree. 
Moreover, since clearly in the small step machine every step depends
only on the previous one, the definition of the translation $(-)^{\tt s}$ 
and Lemma \ref{lem:left-depth-first}
imply that also in $\mathrm{K}_{\mathcal{B}}^{\mathcal{C}}$ every execution step depends only on the previous one.
\begin{example}
By returning to the computation example in Table \ref{TableCompEx}, it is worth 
noting that to pass from the configuration $\phi$ to the configuration 
$\psi$ all necessary information are already present in the configuration 
$\phi$ itself. We can view such a step as a $\reddelta$ step
$(\lif{\0}{\x_1}{\x_1})^{\mathcal{A}_{3}}\reddelta (\x_1)^{\mathcal{A}_3}$
noting that $(\x_1)^{\mathcal{A}_3}\equiv(\x_1)^{\mathcal{A}_2}$.
\end{example}
In fact, the behaviour shown in the above example can be generalized, so in this sense we don't need neither 
mechanism for backtracking nor the memorization of parts of the computation 
tree.
Using this property, we can define in a similar way the notion of space used to evaluate a term in the two machines.

%% file: spaceMeasures.tex
Let us first define the \emph{size} of a configuration in both the machines.
\begin{definition}$ $
\begin{enumerate}
\item If $\langle \mathcal{S},\mathcal{C},\mathcal{A}\succ \M\rangle$ is a configuration in 
$\mathrm{k}_{\mathcal{B}}^{\mathcal{C}}$, then its {\em size} is $|\mathcal{S}|+|\mathcal{C}|+|\mathcal{A}|+|\M|$.
\item If $\phi\Yright \mathcal{C},\mathcal{A}\models \M\ev{\tt b}$ is a configuration in 
$\mathrm{K}_{\mathcal{B}}^{\mathcal{C}}$, then its {\em size} (denoted by
$| \phi |$) is $|\mathcal{C}|+|\mathcal{A}|+|\M|$.
\end{enumerate}
\end{definition}
We can now define the \emph{required space} in both the machines 
as the maximal size of a configuration in the computation.
\begin{definition}$ $
\begin{enumerate}
\item Let $\langle \varepsilon, [\circ],\epsilon\succ \M\rangle \mapsto^{*} \tt b$ 
be a computation in $\mathrm{k}_{\mathcal{B}}^{\mathcal{C}}$. Then
its \emph{required space}, denoted by ${\tt space}_{s}(\M)$, is 
the maximal size of a configuration in it. 
\item 
Let $\nabla:: [\circ],\epsilon\models \M\ev{\tt b}$ be a computation in 
$\mathrm{K}_{\mathcal{B}}^{\mathcal{C}}$. Then
its \emph{required space}, denoted by ${\tt space}(\M)$, is the maximal 
size of a configuration in $\nabla$. 
\end{enumerate}
\end{definition}
We can now show that the relation on the required space of the two machines is the expected one.
\begin{lemma}Let $\M\in\prog$. Then:
$${\tt space}_{s}(\M)\leq{\tt space}(\M)$$
\end{lemma}
\begin{proof}
 By definition of the translation $(-)^{\tt s}$ and Lemma \ref{lem:left-depth-first}.
\end{proof}
So from now on we can restrict our attention 
to prove the polynomial space measure soundness in the case of the big step evaluation machine.


\subsection{Space Measures}
In this subsection we will connect the space measure of the big step machine with the one of the typable term to be evaluated. In particular, we emphasize the relations between machine computations and type derivations.\\
In what follows we introduce some relations between 
the size of the contexts and the behaviour of the machine, which will be useful later.
\begin{definition}Let $\nabla$ be a computation and $\phi\in \nabla$ a configuration. Then:
\begin{itemize}
\item the symbol $\#_{\beta}(\phi)$ denotes the number of applications of the $(\beta)$ 
rule in $\pathi(\phi)$,
\item the symbol $\#_{h}(\phi)$ denotes the number of applications of the $(h)$ rule in $\pathi(\phi)$, 
\item the symbol $\#_{\tt if}(\phi)$ denotes the number of applications of $(\tt if\ \0)$ and
$(\tt if \ \1)$ rules in $\pathi(\phi)$. 
\end{itemize}
\end{definition}
The cardinality of the contexts in a configuration $\phi$ is a measure of the number of some rules performed by the machine in the path to reach $\phi$.
\begin{lemma}
\label{size_number}
Let $\nabla \mathop{::}\models \M\ev {\tt b}$ be a computation. Then, for each configuration
$\phi \Yright \mathcal{C}_{i},\mathcal{A}_i\models {\tt P}_i\ev {\tt b}'\in \nabla$:
 \begin{enumerate}
\item  $\#(\mathcal{A}_i)= \#_{\beta}(\phi)$
\item  $\#(\mathcal{C}_i)= \#_{\tt if}(\phi)$ 
\end{enumerate}
\end{lemma}
\begin{proof}$ $
\begin{enumerate}
\item Easy, by induction on the length of $\pathi(\phi)$,
since m-contexts can grow only by applications 
of the $(\beta)$ rule.
\item Easy, by induction on the length of $\pathi(\phi)$,
since $\bt$-contexts can grow only by applications 
of $(\ifz\0)$ and $(\ifz\1)$ rules. \qed
\end{enumerate}
\end{proof}
The following is a key property for proving soundness.
\begin{property} 
\label{property_subterms}
Let $\M\in \prog$ and 
$\nabla \mathop{::}\models \M\ev {\tt b}$ then for each
$\phi\Yright \mathcal{C},\mathcal{A}\models {\tt P}\ev
{\tt b}'\in \nabla$ if 
$\{x_j:=\N_j\}\in\mathcal{A}$ then $\N_j$ is
an  instance (possibly with fresh variables) of a subterm of $\M$.
\end{property}
\begin{proof}
The property is proven by contradiction. Take the configuration $\psi$ with 
minimal path 
from it to the root of $\nabla$, such that in its m-context $\mathcal{A_\psi}$
there is $\x_j:=\N_j$, where $\N_j$ is not an instance of a 
subterm of $\M$. Let $p$ be the length of this path. 
Since the only rule
that makes the m-context grow is a $(\beta)$ rule we are in a situation like
the following:
$$
\infer{
\mathcal{C},\mathcal{A}'\models (\lambda \x.{\tt P})\N_j {\tt V}_1\cdots{\tt V}_n \ev {\tt b}}{
\psi\Yright\mathcal{C},\mathcal{A}'@\{\x_j:=\N_j\}\models { \tt P}[\x_j/\x]{\tt V}_1\cdots{\tt V}_n \ev {\tt b}}
$$
If $\N_j$ is not an instance of a subterm of $\M$ it has been obtained
by a substitution.
Substitutions can be made only through applications of rule $(h)$
replacing the head variable.
Hence, by the shape of  $(\lambda \x.{\tt P})\N_j {\tt V}_1\cdots{\tt V}_n$,
the only possible situation is that there exists an application of rule $(h)$ 
as:
$$
\infer{
\mathcal{C},\mathcal{A}'\models \y{\tt V'}_1\cdots{\tt V'}_m \ev {\tt b}}
{
[\y:={\tt M'}]\in \mathcal{A}'& \mathcal{C},\mathcal{A}'\models {\tt M}' {\tt V'}_1\cdots{\tt V'}_m \ev {\tt b}}
$$
with $\N_j$ a subterm of $\tt M'$. But this implies $\tt M'$ is not an instance
of a  subterm of $\M$ 
and it has been introduced by a rule 
of a path of length less than 
$p$, contradicting the hypothesis.
\end{proof}
The next lemma gives upper bounds to the size of the m-context, 
of the $\bt$-context and
of the subject of a configuration.\\
\begin{lemma}
\label{m-context-size}
 Let $\M\in \prog$ and 
$\nabla ::\models \M\ev {\tt b}$ then for each configuration
$\phi\Yright\mathcal{C},\mathcal{A}\models {\tt P}\ev {\tt b}'\in \Pi$: 
\begin{enumerate}
\item $|\mathcal{A}|\leq \#_{\beta}(\phi)(|\M|+1)$
\item $|{\tt P}|\leq (\#_{h}(\phi)+1)|\M|$
\item $|\mathcal{C}|\leq \#_{\tt if}(\phi)(\max\{|\N| \ |\  
\psi \Yright\mathcal{C}',\mathcal{A}'\models {\tt N}\ev {\tt b}''\in  \pathi(\phi)
 \})$
\end{enumerate}
\end{lemma}
\begin{proof}\
\begin{enumerate}
\item  By inspection of the  rules of Table \ref{BigStepB} 
it is easy to verify that m-contexts can grow only by applications 
of the $(\beta)$ rule.
So the conclusion follows by Lemma \ref{size_number}.1 and Property \ref{property_subterms}.
\item By inspection of the rules of Table \ref{BigStepB} 
it is easy to verify that the subject can grow only by substitutions through 
applications of the $(h)$ rule. So 
the conclusion follows by Property \ref{property_subterms}.
\item By inspection of the rules of Table \ref{BigStepB} 
it is easy to verify that $\bt$-contexts can grow only by applications 
of $(\tt if \ \0)$ and
$(\tt if \ \1)$ rules.
So the conclusion follows directly by Lemma \ref{size_number}.2. \qed
\end{enumerate}
\end{proof}


%% file: correctness.tex
\section{PSPACE Soundness}
\label{correctness}
In this section we will show that $\BSTA$ is correct for polynomial space 
computations.
The degree of a type derivation, i.e. the maximal nesting of applications
of the rule $(sp)$ in it, is the key notion in order to obtain the 
correctness. In fact, we will prove that each program typable through 
a derivation with degree $d$ can be executed on the machine
$\mathrm{K}_{\mathcal{B}}^{\mathcal{C}}$ in space polynomial in its size, where
the maximum exponent of the polynomial is $d$. 
So, by considering fixed degrees we get PSPACE soundness. 
Considering a fixed $d$ is not a limitation. Indeed 
until now, in $\BSTA$ programs we have not distinguished between the program code and 
input data. But it will be shown in Section \ref{completeness}
that data types are typable through derivations with degree 0. Hence,
the degree can be considered as a real characteristic of the program code.\\
Moreover, every $\BSTA$ program can be typed through 
derivations with different degrees, nevertheless for each program
there is a sort of minimal derivation for it, with respect to the degree.
So, we can stratify programs with respect to the degree of their derivations,
according to the following definition.
\begin{definition}\
\begin{enumerate}
\item
Let $\Pi$ be a type derivation. The \emph{degree} of $\Pi$, denoted ${\tt d}(\Pi)$ 
is the maximal nesting of 
applications of rule $(sp)$ in $\Pi$. 
It is inductively defined on the height of $\Pi$ as follows:
\begin{itemize}
\item if $\Pi$ consists of a $(Ax)$ or of a $(\bt_{\tt b}I)$  rule
then ${\tt d}(\Pi)=0$
\item if $\Pi$ ends by a rule 
$$
\infer[(R)]{\Gamma\vdash \M:\sigma}{\Sigma}
$$
where $(R)\in\{(w),(\lin I),(m),(\forall E),(\forall I)\}$ then
${\tt d}(\Pi)={\tt d}(\Sigma)$ 
\item 
If $\Pi$ ends by a rule
$$
\infer[(\lin E)]{\Gamma,\Delta\vdash \M\N:A}{\Sigma\pder\Gamma\vdash \M:\sigma\lin A & \Theta\pder\Delta \vdash \N:\sigma}
$$
then 
${\tt d}(\Pi)=\max\{{\tt d}(\Sigma),{\tt d}(\Theta) \}$
\item 
If $\Pi$ ends by a rule
$$
\infer[(\bt E)]{\Gamma\vdash \lif{\M}{\N_\0}{\N_\1}:\sigma}{\Sigma\pder\Gamma\vdash \M:\bt & \Theta_\0\pder\Gamma \vdash \N_\0:\sigma & \Theta_\1\pder\Gamma\vdash \N_\1:\sigma}
$$
then 
${\tt d}(\Pi)=\max\{{\tt d}(\Sigma),{\tt d}(\Theta_\0),{\tt d}(\Theta_\1) \}$
\item if $\Pi$ ends by a rule 
$$
\infer[(sp)]{!\Gamma\vdash \M:!\sigma}{\Sigma\pder \Gamma\vdash \M:\sigma}
$$
then
${\tt d}(\Pi)={\tt d}(\Sigma)+1$ 
\end{itemize}
\item
For each $d\in\mathbb{N}$ the set $\prog_d$ is the set of $\BSTA$ programs typable through derivation
with degree $d$. $$\prog_d=\{ \M \mid\ \Pi\pder \vdash \M:\bt\ \land\
{\tt d}(\Pi)=d\}$$
\end{enumerate}
\end{definition}
Clearly $\mathcal{P}$ corresponds to the union for $n\in \mathbb{N}$
of the different $\mathcal{P}_n$. Moreover if $\M\in \mathcal{P}_d$ then 
$\M\in \mathcal{P}_e$ for every $e\geq d$.\\

This section is divided into two subsections. In the first, we will prove an intermediate result, namely
we will give the notion of space weight of a derivation, and we will prove that
the subject reduction does not increase it. 
Moreover, this result is extended to the machine $\mathrm{K}_{\mathcal{B}}^{\mathcal{C}}$.
In the second subsection, the 
soundness with respect to PSPACE will be proved. 
\subsection{Space and $\BSTA$}
We need to define measures
of both terms and proofs, which are an adaptation of those given by Lafont in
\cite{Lafont04}.
\begin{definition}$ $
\begin{itemize}
\item 
The \emph{rank} of a rule $(m)$:
$$\infer[(m)]{\A, {\tt x}:!\sigma\der {\tt M}[{\tt x}/{\tt x}_1,\cdots,{\tt x}/{\tt x}_n ]:\mu}
{\A,{\tt x}_1:\sigma,\ldots,{\tt x}_n:\sigma\der {\tt M}:\mu}$$ 
is the number $k\leq n$ of variables $\x_i$ such that $\x_i$ belongs to the free
variables of $\M$.
Let $r$ be the
the maximal 
rank of a rule $(m)$ in $\Pi$.
Then, the rank of $\Pi$ is ${\tt rk}(\Pi)=\max(r,1)$.
\item 
Let $r$ be a natural number. The \emph{space weight} $\delta(\Pi,r)$ of $\Pi$ with respect
to $r$ is defined inductively as follows:
\begin{itemize}
\item If $\Pi$ consists of a $(Ax)$ or of a $(\bt_{\tt b} I)$ rule, then 
$\delta(\Pi,r)=1$.
\item
If $\Pi$ ends by a rule
$$
\infer[(\lin I)]{\Gamma \der \lambda \x.\M:\sigma \lin A }
{\Sigma\pder \Gamma,\x:\sigma\der \M:A }
$$ 
then $\delta(\Pi,r)=\delta(\Sigma,r)+1$.
\item
If $\Pi$ ends by a rule
$$
\infer[(sp) ]{!\Gamma \der \M:!\sigma }{\Sigma\pder \Gamma \der \M:\sigma  } 
$$
then $\delta(\Pi,r)=r \delta(\Sigma,r)$.
\item
If $\Pi$ ends by a rule
$$
\infer[(\lin E)]{\Gamma,\Delta \der \M\N:A }{\Sigma\pder \Gamma \der \M:\mu\lin A &
\Theta\pder \Delta\der \N:\mu}
$$
then $\delta(\Pi,r)=\delta(\Sigma,r)+  \delta(\Theta,r)+1$.
\item 
If $\Pi$ ends by a rule
$$
\infer{\A \der \lif{{\tt M}}{{\tt N_{\0}}}{{\tt N_{\1}}}:A}
{\Sigma\pder \A\der {\tt M}:\bt &\Theta_{\0}\pder \A\der {\tt N_{\0}}:A  & \Theta_{\1}\pder \A\der {\tt N_{\1}}:A}
$$
then $\delta(\Pi,r)=\max \{\delta(\Sigma,r),\delta(\Theta_{\0},r),\delta(\Theta_{\1},r)\}+1$
\item
In any other case $\delta(\Pi,r)=\delta(\Sigma,r)$ where $\Sigma$
is the unique premise derivation.
\end{itemize}
\end{itemize}
\end{definition}
In order to prove that the subject reduction does not increase the space weight of a derivation,
we need to rephrase the Substitution Lemma taking into account this measure. 
\begin{lemma}[Weighted Substitution Lemma]
\label{weightedSubstitution}
Let $\Pi\pder \A, \x:\mu \der \M:\sigma$ and $\Sigma\pder \B \der \N:\mu$ 
such that $\A\#\B$. 
There exists $\Theta \pder \A,\B \der \M[\N/\x]:\sigma$ such that if 
$r\geq {\tt rk}(\Pi)$:
$$\delta(\Theta,r)\leq \delta(\Pi,r)+\delta(\Sigma,r)$$
\end{lemma}
\begin{proof}
It suffices to verify how the weights are modified by the proof
of Lemma \ref{lem:substitution}. We proceed 
by induction on the height of $\x$ in $\Pi$.
Base cases are trivial and in the cases where $\Pi$ ends by 
$(\lin I),(\forall I), (\forall E)$ and $(\lin E)$ rules the conclusion follows 
directly by induction hypothesis.\\
Consider the case $\Pi$ ends by: 
$$\infer[(sp)]{\A, \x:\mu \der \M:\sigma}{\Pi'\pder \Gamma', \x:\mu' \der \M:\sigma'}$$ 
Then 
by Lemma \ref{lem:generation}.3 $\Sigma\leadsto \Sigma ''$ which is composed 
by 
a subderivation ending with an $(sp)$ rule with premise $\Sigma'\pder \B'\der \N:\mu' $ followed by a sequence of rules $(w)$ and/or $(m)$. By induction 
hypothesis we have a derivation $\Theta'\pder \A',\B'\der \M[\N/\x]:\sigma' $.
By applying the rule $(sp)$ and the sequence of $(w)$ and/or $(m)$ rules we obtain 
$\Theta\pder \A,\B\der \M[\N/\x]:\sigma$. Now, $\delta(\Pi,r)=r\delta(\Pi',r)$ and 
$\delta(\Sigma,r)=r\delta(\Sigma',r)$.
By the induction hypothesis
$\delta(\Theta',r)\leq \delta(\Pi',r)+\delta(\Sigma',r)$
and applying $(sp)$:
$$\delta(\Theta,r)\leq r(\delta(\Pi',r)+\delta(\Sigma',r))=\delta(\Pi,r)+\delta(\Sigma,r)$$
Consider the case $\Pi$ ends by:
$$
\infer[(\bt E)]{\A,\x:\mu \der \lif{{\M_{0} }}{{\M_{1}}}{{\M_{2}}}:A}{\Pi_{0}\pder \A,\x:\mu\der {\M_{0}}:\bt &  \Pi_{1}\pder \A,\x:\mu\der {\M_{1}}:A  &\Pi_{2}\pder \A,\x:\mu\der {\M_{2}}:A}
$$
Then, by the induction hypothesis there are derivations
$\Theta_0\pder \Gamma,\Delta \der\M_{0}[\N/\x]:\bt$, 
$\Theta_1\pder \Gamma,\Delta \der\M_{1}[\N/\x]:A$ and 
$\Theta_2\pder \Gamma,\Delta \der\M_{2}[\N/\x]:A$
such that 
$\delta(\Theta_i,r)\leq \delta(\Pi_i,r)+\delta(\Sigma,r)$ 
for $0\leq i\leq 2$.
By applying a $(\bt E)$ rule we obtain a derivation $\Theta$ with conclusion:
$$\Gamma,\B \der \lif{{\M_{0}[\N/\x] }}{{\M_{1}[\N/\x]}}{{\M_{2}[\N/\x]}}:A$$
Since, by definition $\delta(\Pi,r)=\max_{0\leq i\leq 2}(\delta(\Pi_i,r))+1$, we have
$$\delta(\Theta,r)\leq\max_{0\leq i\leq 2}(\delta(\Pi_i,r)+\delta(\Sigma,r))+1= \max_{0\leq i\leq 2}(\delta(\Pi_i,r))+1+\delta(\Sigma,r)=\delta(\Pi,r)+\delta(\Sigma,r)$$
Consider the case $\mu\equiv !\mu'$ and  $\Pi$ ends by: 
$$
\infer[(m)]{\A, {\tt x}:!\mu'\der {\tt M}[{\tt x}/{\tt x}_1,\cdots,{\tt x}/{\tt x}_m ]:\sigma}{\Pi'\pder \A,{\tt x}_1:\mu',\ldots,{\tt x}_m:\mu'\der {\tt M}:\sigma}
$$
By Lemma \ref{lem:generation}.3, $\Sigma\leadsto \Sigma ''$ ending 
by an 
$(sp)$ rule with premise $\Sigma'\pder \B'\der \N:\mu' $ followed
by a sequence of  rules $(w)$ and/or $(m)$. 
Hence, $\delta(\Sigma,r)=r\delta(\Sigma',r)$.
Consider fresh copies of the derivation
$\Sigma'$ i.e. $\Sigma'_{j}\pder \B_{j}'\der \N_{j}:\mu' $
 where $\N_j$ and $\B_j'$ 
 are fresh copies of $\N$ and $\B'$ respectively, trivially 
$\delta(\Sigma',r)=\delta(\Sigma'_j,r)$ ($1 \leq j \leq m$).
\\
Let $\x_i$ be such that its height is maximal between the heights 
of all $\x_j$ ($1\leq j\leq m$).
By induction hypothesis there is a  derivation: 
$$\Theta_i\pder \A,\x_1:\mu',\ldots,\x_{i-1}:\mu',\x_{i+1}:\mu',\ldots,\x_m:\mu',\B'_i\der \M[\N_i/\x_i]:\sigma $$
and since $\delta(\Pi,r)=\delta(\Pi',r)$, we have 
$\delta(\Theta_i,r)\leq \delta(\Pi',r)+\delta(\Sigma',r)$.
Then, we can repeatedly 
apply induction  hypothesis to obtain a derivation
$\Theta'\pder \A,\B'_1,\ldots ,\B'_m\der \M[\N_1/\x_1,\cdots,\N_m/\x_m]:\sigma$.
such that 
$\delta(\Theta',r)\leq \delta(\Pi',r)+m\delta(\Sigma',r)$ and since
$r\geq {\tt rk}(\Pi)$ then:
$$\delta(\Theta',r)\leq \delta(\Pi',r)+r\delta(\Sigma',r)
=\delta(\Pi,r)+\delta(\Sigma,r)$$
Finally
 by applying  repeatedly  the rules $(m)$ and $(w)$ 
that leave
the space weight $\delta$ unchanged,
the conclusion follows.
\end{proof}

We are now ready to show that the space weight $\delta$ gives a bound on 
the number of both $\beta$ and $\tt if$ rules in a computation path of the machine 
$\mathrm{K}_{\mathcal{B}}^{\mathcal{C}}$.

\begin{lemma}
\label{lambda-weight-decrease}
Let ${\tt P}\in \prog$ and $\nabla\mathop{::}\models {\tt P}\ev {\tt b}$. 
\begin{enumerate}
\item Consider an occurrence in $\nabla$ of the rule:
$$
\infer[(\beta)]{\mathcal{C},\mathcal{A}\models (\lambda \x. \M)\N{\tt V}_1\cdots {\tt V}_m\Downarrow {\tt b}}{
\mathcal{C},\mathcal{A}@ \{\x':=\N\}\models \M[\x'/\x] {\tt V}_1\cdots {\tt V}_m\Downarrow {\tt b} }
$$
Then, for every derivation 
$\Sigma\pder \der ((\lambda \x. \M)\N{\tt V}_1\cdots {\tt V}_m)^{\mathcal{A}}:\bt$
there exists a derivation 
$\Theta\pder \der (\M[\x'/\x] {\tt V}_1\cdots {\tt V}_m)^{\mathcal{A}@ \{\x':=\N\}}:\bt$
such that for every $r\geq{\tt rk}(\Sigma)$: $$\delta(\Sigma,r)>\delta(\Theta,r)$$
\item Consider an occurrence in $\nabla$ of 
an $\tt if$ rule as:
$$
\infer[(\ifz\ {\tt b})]{\mathcal{C},\mathcal{A}\models(\lif{\M}{\N_{\0}}{\N_{\1}}){\tt V}_1\cdots {\tt V}_m \Downarrow \tt b'}
{\mathcal{C}',\mathcal{A}\models \M \Downarrow {\tt b}
& 
\mathcal{C},\mathcal{A}\models \N_{\tt b}{\tt V}_1\cdots{\tt V}_m \Downarrow {\tt b'}
}
$$
where $\mathcal{C}'\equiv \mathcal{C}[(\lif{[\circ]}{\N_{\0}}{\N_{\1}}){\tt V}_1\cdots {\tt V}_m]$. 
Then, for each derivation
$\Sigma\pder \der ((\lif{\M}{\N_{\0}}{\N_{\1}}){\tt V}_1\cdots {\tt V}_m )^{\mathcal{A}}:\bt$
there are derivations
$\Theta\pder \der (\M)^{\mathcal{A}}:\bt$ and
$\Pi\pder \der (\N_{\tt b}{\tt V}_1\cdots {\tt V}_m )^{\mathcal{A}}:\bt$
such that for every $r\geq{\tt rk}(\Sigma)$: 
$$\delta(\Sigma,r)>\delta(\Theta,r)\quad \textrm{and}\quad 
\delta(\Sigma,r)>\delta(\Pi,r)$$
\end{enumerate}
\end{lemma}
\begin{proof}\ 
\begin{enumerate}
\item We proceed by induction on $m$.
Consider the case $m=0$. We need to prove that if
$\Pi\pder \A \der (\lambda \x.\M)\N:\sigma$, then there exists 
$\Pi'\pder \A\der \M[\N/\x]:\sigma $ with 
${\tt rk}(\Pi)\geq {\tt rk}(\Pi')$
such that for $r\geq {\tt rk}(\Pi)$:
$$
\delta (\Pi, r)> \delta (\Pi',r )
$$
Since $(\forall R),(\forall L),(m)$ and $(w)$ rules do not change 
the space weight $\delta$, without loss of generality 
we can assume that $\Pi$ ends as follows:
$$
\infer[(sp)^{n}]{!^{n}\A_1,!^{n}\A_2\der {\tt (\lambda \x. M)N}:!^{n}A}{
  \infer[(\lin E)]{\A_1,\A_2\der {\tt (\lambda x.M)N}:A}{
    \infer[(\lin I)]{\A_1\der {\tt \lambda x.M}:\sigma\lin A}{
      \Pi_1\pder \A_1, {\tt x}:\sigma\der {\tt M}:A  
    }
    & 
    \Pi_2\pder  \A_2 \der {\tt N}:\sigma 
  }
}
$$
where $\A_1\#\A_2$, $\A=!^{n}\A_1,!^{n}\A_2$, $\sigma\equiv !^{n}A$ for $n\geq 0$. 
Clearly, by definition of the space weight $\delta$, 
we have
$\delta(\Pi,r)=r^{n}(\delta(\Pi_1,r)+1 +\delta(\Pi_2,r))$.
Since $r\geq \rk(\Pi)\geq \rk(\Pi_1)$, by Lemma \ref{weightedSubstitution} there exists a derivation 
$\Pi_3 \pder \A\der \M[\N/\x]:A $
such that 
$
\delta (\Pi_3, r)\leq \delta(\Pi_1,r)+\delta(\Pi_2,r)
$. 
Hence, we can construct $\Pi'$ ending as:
$$
\infer[(sp)^{n}]{!^{n}\A_1,!^{n}\A_2\der {\tt \M[\N/\x]}:!^{n}A}{
  \Pi_3\pder \A,\B\der \M[\N/\x]:A
}
$$
Clearly $\delta (\Pi', r)\leq r^{n}(\delta(\Pi_1,r)+\delta(\Pi_2,r))<\delta(\Pi,r)$ and
so the conclusion follows.\\
The inductive step $m=k+1$ follows easily by the induction hypothesis.

\item It follows directly by the definition of the space weight 
$\delta$. \qed
 \end{enumerate}
 \end{proof}
It is easy to verify that $(h)$ rules leave the space weight 
unchanged, since $({\tt x}{\tt V}_1\cdots{\tt V}_m)^{\mathcal{A}}\equiv({\tt N}{\tt V}_1\cdots{\tt V}_m)^{\mathcal{A}}$ if $\{{\tt x}:={\tt N}\}\in \mathcal{A}$. Hence, a direct consequence of the above lemma is the following.
\begin{lemma}
\label{rulesnumber}
Let $\Pi\pder \M:\bt$ and $\nabla\mathop{::}\models \M\ev {\tt b}$. Then for each 
$\phi\in\nabla$ such that  $\phi\Yright \mathcal{C},\mathcal{A}\models\N\ev {\tt b}'$ if $r\geq \rk(\Pi)$:
$$\#_{\beta}(\phi)+\#_{\tt if}(\phi)\leq \delta(\Pi,r)$$
\end{lemma}
\begin{proof}
Easy, by Lemma \ref{lambda-weight-decrease}.
\end{proof}
Subject reduction does not increase the space weight.
\begin{property}
\label{prop:weightedSubjectReduction}
Let $\Pi\pder \A \der \M:\sigma$ and $\M\red^{*} \N$. Then there exists 
$\Pi'\pder \A\der \N:\sigma $ with ${\tt rk}(\Pi)\geq {\tt rk}(\Pi')$
such that for each $r\geq {\tt rk}(\Pi)$:
$$
\delta (\Pi, r)\geq \delta (\Pi',r )
$$
\end{property}
\begin{proof}
By Lemma \ref{weightedSubstitution} and definition of $\delta$.
\end{proof}
It is worth noticing that a reduction inside an $\tt if$ does not necessarily decrease the space weight $\delta$. This is the reason why we
consider a non-strict inequality in the statement of the above property. \\
The previous result can be extended to the machine $\mathrm{K}_{\mathcal{B}}^{\mathcal{C}}$
in the following way.
\begin{property}
Let $\Pi\pder \der \M:\bt$ and $\nabla\mathop{::} \models \M\ev {\tt b}$. For each configuration
$\phi\in \nabla$ such that $\phi\Yright \mathcal{C},\mathcal{A}\models \N\ev {\tt b}'$ 
and $\mathcal{C}\not\equiv \circ$
there exist derivations
$\Sigma\pder \vdash (\mathcal{C}[\N])^{\mathcal{A}}:\bt$ and
 $\Theta\pder \vdash (\N)^{\mathcal{A}}:\bt$ such that  $\Theta$ is a proper subderivation 
of $\Sigma$
 and for each  $r\geq \rk(\Pi)$:
$$\delta(\Pi,r)\geq\delta(\Sigma,r)> \delta(\Theta,r)$$ 
\end{property}
\begin{proof}
 Easy. 
\end{proof}
Note that in the above property we ask for $\mathcal{C}\neq \circ$ just in order to make the second inequality strict.
\subsection{Proof of PSPACE Soundness}
As defined in the previous section, the space used by the machine $\mathrm{K}_{\mathcal{B}}^{\mathcal{C}}$
is the maximum space used by its configurations. In order to give an account of this space, we need to measure
the increasing of the size of a term during its evaluation. 
The key notion for realizing this measure is that of
sliced occurrence of a variable, which takes into account that in performing an $\tt if$ reduction
a subterm of the subject is erased.
In particular, by giving a bound on the number of sliced occurrences of variables we 
obtain a bound on the number of applications of the $h$ rule in a path.
\begin{definition}
The number of \emph{sliced occurrences} $n_{so}(\x,\M)$ of  the variable
$\x$ in $\M$ is defined as:
\begin{center}
\begin{tabular}{c}
$\ n_{so}(\x,\x)=1,\ n_{so}(\x,{\tt y})=
n_{so}(\x,\0)= 
n_{so}(\x,\1)=0,\ $\\
\\
$n_{so}(\x,{\tt MN})=n_{so}(\x,\M)+n_{so}(\x,\N), \ $
$n_{so}(\x,\lambda{\tt y}.\M)= n_{so}(\x,\M),\ $ \\
\\
$n_{so}(\x,\lif{\M}{\N_{\0}}{\N_{\1}})=\max\{ n_{so}(\x,\M),
n_{so}(\x,\N_{\0}),n_{so}(\x,\N_{\1})\}
$
\end{tabular}
\end{center}
\end{definition}
A type derivation gives us some information about the number of sliced occurrences of a 
free variable $\x$ in its subject $\M$.
\begin{lemma}
\label{number-sliced-occurrences}
Let $\Pi\pder \A, \x:!^{n} A\der \M:\sigma$ then $n_{so}(\x,\M)\leq {\tt rk}(\Pi)^{n}$.
\end{lemma}
\begin{proof}\
By induction on $n$. \\
Case $n=0$. The conclusion follows easily by 
induction on $\Pi$. 
Base cases are trivial. In the case 
$\Pi$ ends by $(\bt E)$, the conclusion follows by $n_{so}(\x,\M)$ definition 
and induction hypothesis. The other cases follow directly from the induction 
hypothesis remembering the side condition $\Gamma\#\Delta$ in $(\lin E)$ case.\\
Case $n>0$. By induction on $\Pi$. Base case is trivial. 
Let the
last rule of $\Pi$ be:
$$
\infer[(\bt E)]{\A \der \lif{{\tt M'}}{{\tt N_{\0}}}{{\tt N_{\1}}}:B}{\Sigma\pder \A\der {\tt M'}:\bt &  \Theta_{\0}\pder \A\der {\tt N_{\0}}:B  &\Theta_{\1}\pder \A\der {\tt N_{\1}}:B}
$$
where $x:!^nA \in \A$.  By induction hypothesis 
$n_{so}(\x,\M')\leq {\tt rk}(\Sigma)^{n}$
and 
$n_{so}(\x,\N_{i})\leq {\tt rk}(\Theta_{i})^{n}$ for $i\in\{\0,\1\}$.
By definition of rank ${\tt rk}(\Pi)=\max\{{\tt rk}(\Sigma),{\tt rk}(\Theta_{\0}),{\tt rk}(\Theta_{\1})\}$
and since by definition
$n_{so}(\x,\lif{\M}{\N_{\0}}{\N_{\1}})$ is equal to $\max\{ n_{so}(\x,\M),
n_{so}(\x,\N_{\0}),n_{so}(\x,\N_{\1})\}
$, then the conclusion follows.\\
Let the
last rule of $\Pi$ be:
$$
\infer[(m)]{\A, {\tt x}:!^{n}A\der {\tt N}[{\tt x}/{\tt x}_1,\cdots,{\tt x}/{\tt x}_m ]:\mu}{\Sigma\pder\A,{\tt x}_1:!^{n-1}A,\ldots,{\tt x}_m:!^{n-1}A\der {\tt N}:\mu}
$$
where ${\tt N}[{\tt x}/{\tt x}_1,\cdots,{\tt x}/{\tt x}_m ]\equiv {\tt M}$.
By induction hypothesis  $n_{so}(\x_i,\N)\leq {\tt rk}(\Sigma)^{n-1}$ for
$1\leq i\leq m$ and since ${\tt rk}(\Sigma)\leq {\tt rk}(\Pi)$ 
the conclusion follows easily.
In every other case the conclusion follows directly by induction hypothesis.
\end{proof}
It is worth noting that the above lemma and the subject reduction 
property gives dynamical informations about the number of sliced 
occurrences of a variable. 
\begin{lemma}
\label{dynamicSliced}
Let $\Pi\pder \A, \x:!^{n}A \der \M:\sigma$ and $\M\red \N$. Then, 
$n_{so}(\x,\N)\leq {\tt rk}(\Pi)^{n}$.
\end{lemma}
\begin{proof}
  Easy, by Property \ref{prop:weightedSubjectReduction} and Lemma \ref{number-sliced-occurrences}.
\end{proof}
The lemma above is essential to prove the following important property.
\begin{lemma}
\label{hnumber}
Let $\M\in \prog_{d}$ and 
$\nabla \mathop{::}\models \M\ev {\tt b}$ then for each
$\phi\Yright\mathcal{C},\mathcal{A}\models {\tt P}\ev {\tt b}'\in \nabla$: 
$$\#_{h}(\phi)\leq\#(\mathcal{A})|\M|^{d}$$
\end{lemma}
\begin{proof}
For each $[{\tt x'}:=\N]\in \mathcal{A}$ the variable $\tt \x'$ is a fresh copy
of a variable ${\tt x}$ originally bound in $\M$.  Hence, 
$\M$ contains a subterm $(\lambda{\tt x}.{\tt P}){\tt Q}$ and
there exists a derivation $\Pi$ such that 
$\Pi\pder  {\tt x}:!^{n}A\der {\tt P}:B$.\\
By Lemma
\ref{dynamicSliced} for every $\tt P'$ such that $\tt P\red^{*} P'$
we have   $n_{so}(\x,{\tt P}')\leq {\tt rk}(\Pi)^{n}$.
So, in particular the number of applications of $h$ rules on the 
variable $\tt \x'$ is bounded by ${\tt rk}(\Pi)^{n}$. 
Since $|\M|\geq \rk(\Pi)$ and $d\geq n$, the conclusion follows.
\end{proof}
The following lemma relates the space weight with both the size of the term and the degree
of the derivation. 
\begin{lemma} \label{delta-properties}
Let $\Pi\pder \Gamma \der \M:\sigma$.
\begin{enumerate}
\item $\delta(\Pi,1)\leq |\M|$
\item $\delta(\Pi,r)\leq \delta(\Pi,1)\times r^{{\tt d}(\Pi)}$
\item $\delta(\Pi,{\tt rk}(\Pi))\leq |\M|^{{\tt d}(\Pi)+1}$
\end{enumerate}
\end{lemma}
\begin{proof}$ $
\begin{enumerate}
\item By induction on $\Pi$. Base cases are trivial. Cases 
$(sp),(m),(w),(\forall I)$ and $(\forall E)$ follow directly by induction 
hypothesis. The other cases follow by definition of $\delta(\Pi,1)$.
\item By induction on $\Pi$. Base cases are trivial. Cases 
$(m),(w),(\forall I)$ and $(\forall E)$ follow directly by induction 
hypothesis. The other cases follow by induction hypothesis and the definitions of $\delta(\Pi,r)$ and ${\tt d}(\Pi)$.
\item By definition of rank it is easy to verify that 
${\tt rk}(\Pi)\leq |\M|$, hence by the previous two points the conclusion follows. \qed
\end{enumerate}
\end{proof}

The next lemma gives a bound on the dimensions of all the components of a machine configuration, namely the term,
the m-context and the $\bt$-context.
\begin{lemma}\label{grow}
Let $\M\in\prog_{d}$ and
$\nabla \mathop{::}   \models \M\ev {\tt b}$. Then for each
$\phi\Yright\mathcal{C},\mathcal{A}\models \N\ev {\tt b}'\in \nabla$: 
\begin{enumerate}
\item $|\mathcal{A}|\leq 2|\M|^{ d+2}$
\item $|\N|\leq 2|\M|^{2d+2}$
\item $|\mathcal{C}|\leq 2|\M|^{3d+3}$
\end{enumerate}
\end{lemma}
\begin{proof}\
\begin{enumerate}
\item By Lemma \ref{m-context-size}.1, 
Lemma \ref{rulesnumber} and Lemma 
\ref{delta-properties}.3. 
$$|\mathcal{A}|\leq \#_{\beta}(\phi)(|\M|+1)
\leq \delta(\Pi,\rk(\Pi))(|\M|+1)\leq |\M|^{{\tt d}+1}(|\M|+1)
\leq
 2|\M|^{ d+2}$$
\item By Lemma \ref{m-context-size}.2,  Lemma \ref{hnumber}, Lemma \ref{size_number}.1, Lemma \ref{rulesnumber} and Lemma \ref{delta-properties}.3:
$$|\N|\leq(\#_{h}(\phi)+1)|\M|
\leq (\#(\mathcal{A})|\M|^{d} +1)|\M|
\leq
 \#_{\beta}(\phi)|\M|^{d+1}+|\M|\leq
2|\M|^{2d+2}$$
\item By Lemma \ref{m-context-size}.3, 
Lemma \ref{size_number}.2, the previous point of this lemma,   
Lemma \ref{rulesnumber} and Lemma \ref{delta-properties}.3:
\begin{multline*}|\mathcal{C}|\leq
\#_{\tt if}(\phi)(\max\{|\N| \ |\  
\psi \Yright\mathcal{C}',\mathcal{A}'\models {\tt N}\ev {\tt b}''\in  \pathi(\phi)
 \})\\
\leq
\#(\mathcal{C})2|\M|^{2d+2}
\leq
 |\M|^{d+1}2|\M|^{2d+2}\leq
2|\M|^{3d+3}\qed
\end{multline*}
\end{enumerate}
\end{proof}
The PSPACE soundness follows immediately from the definition of ${\tt space}(\nabla)$, for a
machine evaluation $\nabla$, and from the previous lemma.

\begin{theorem}[Polynomial Space Soundness]\ \\
Let $\M\in \prog_d$. Then:
$$
{\tt space}(\M)\leq 6|\M|^{3d+3}
$$
\end{theorem}
\begin{proof}
By definition of $\tt space(M)$ and Lemma \ref{grow}.
\end{proof}


%% file: completeness.tex
\section{PSPACE completeness}
\label{completeness}
A well known result of the seventies states that the class of problem decidable by
a Deterministic Turing Machine (DTM) in space polynomial in the length
of the input coincides with the class of problems
decidable by an Alternating Turing Machine (ATM) \cite{ChKS81}
in time polynomial in the length of the input.
$$
  \text{PSPACE} = \text{APTIME}
$$
We use this result, and we prove that each polynomial time 
 ATM  $\mathcal{M}$ can be simulated by a term typable in $\BSTA$.
In order to do this, we will use a result already obtained 
by two of the authors of this paper  
\cite{GaboardiRonchi07csl,Gaboardi07phd},
namely that $\STA$, the type assignment system for the $\lambda$-calculus on which $\BSTA$ is based, characterizes all the polynomial 
time functions. In particular, we use the same encoding as in 
\cite{GaboardiRonchi07csl,Gaboardi07phd} for the representation of the polynomials. Notice that the data types are coded by means of terms that are typable in a uniform way through derivations of degree $0$.
This approach ensures that the degree of the polynomial space 
bound does 
not depends on the input data.

\subsection*{Some syntactic sugar }
Let $\circ$ denotes composition. In particular $\tt M\circ N$ stands for $\tt \lambda z. M(Nz)$ and 
$\M_1\circ \M_2\circ \cdots\circ \M_n$ stands for 
$\lambda {\tt z}. \M_1(\M_2(\cdots(\M_n{\tt z})))$.\\
Tensor product is definable as 
$\sigma \otimes \tau \df \forall \alpha. (\sigma\lin \tau\lin\alpha)\lin \alpha $.
In particular 
$\tt \langle M,N \rangle$ stands for $\tt \lambda x.xMN$ and 
$\tt \letp{z}{x}{y}{N}$ stands for  $\tt z(\lambda x.\lambda y.N)$.
Note that, since $\BSTA$ is an affine system, tensor product enjoys some properties
of the additive conjunction, as to permit the projections:
as usual $\pi_1(\M)$ stands for $\M(\lambda \x.\lambda \y.\x)$
and $\pi_2(\M)$ stands for $\M(\lambda \x.\lambda \y.\y)$.
The $n$-ary tensor product can be easily defined through the binary one 
and we use $\sigma^n$ to denote $\sigma\otimes \cdots\otimes \sigma$ $n$-times. In the sequel we sometimes consider tensor product modulo associativity.
\subsection*{$\bt$-programmable functions}
We need both to generalize the usual notion of lambda definability, given in 
\cite{Barendregt84}, to different kinds of input data, and to specialize it to our typing system. 
\begin{definition}
Let $f:\mathbb{I}_1\times \cdots \times \mathbb{I}_n\to \mathbb{O}$ be a total
function, let $\mathbf{O},\mathbf{I}_1,\ldots,\mathbf{I}_n\in\mathcal{T}_{\bt}$ and let elements $o\in\mathbb{O}$ and $i_j\in\mathbb{I}_j$, for 
$0\leq j\leq n$, be encoded by terms $\underline{\tt o}$ and 
$\underline{{\tt i}_j}$
such that $\vdash \underline{\tt o}:\mathbf{O}$ and $\vdash \underline{{\tt i}_j}:\mathbf{I}_j$.
\begin{itemize}
\item[(i)]
The function $f$ is $\bt$-\emph{definable} if there is a term $\underline{\tt f}\in \Lambda_{\mathcal{B}}$ such that
$\vdash \underline{\tt f} \underline{{\tt i}_1}\cdots \underline{{\tt i}_n}:\mathbf{O}$ and:
$$
 f i_1\cdots i_n=o \iff
\underline{\tt f} \underline{{\tt i}_1}\cdots \underline{{\tt i}_n}=_{\beta}
\underline{\tt o}
$$
\item[(ii)] Let $\mathbb{O}=\bt$.
The function $f$ is $\bt$-programmable if there is a term  $\underline{\tt f}\in \Lambda_{\mathcal{B}}$ 
such that 
$\underline{\tt f}\underline{{\tt i}_{1}}\ldots \underline{{\tt i}_{n}}\in \prog$ and:
$$f(i_1,\ldots i_n)=b\ \iff\ \models\underline{\tt f}\underline{{\tt i}_{1}}\ldots \underline{{\tt i}_{n}}\ev\tt \underline{b} $$ 
\end{itemize}
\end{definition}
\subsection*{Natural numbers and strings of booleans}
Natural numbers, as usual in the $\lambda$-calculus, are represented by Church numerals, i.e. 
$\underline{\tt n}\df \lambda {\tt s}.\lambda {\tt z }.{\tt s}^{n}({\tt z})$.
Each Church numeral $\underline{\tt n}$ is such that $\vdash \underline{\tt n}:\mathbf{N}_i$ for every $i\geq 1$ where the 
\emph{indexed type} $\mathbf{N}_i$ is defined as:
$$\mathbf{N}_i\df \forall \alpha .!^{i}
(\alpha\lin\alpha)\lin\alpha\lin\alpha$$ 
It is easy to check that $\underline{\tt n}$ is typable by means
of derivations with degree $0$.
We  simply use $\mathbf{N}$ to mean $\mathbf{N}_1$.\\
The standard terms
$
{\tt suc} \df \tt\lambda n. \lambda s.\lambda z.  s(nsz),\
{\tt add} \df \tt\lambda n.\lambda m. \lambda s.\lambda z.  n s(msz)$
 and ${\tt mul} \df \tt\lambda n.\lambda m. \lambda s. n(ms)$,
defining successor, addition and  multiplication,
analogously to what happens in STA,
are typable as:
$
\vdash {\tt suc}  : \mathbf{N}_i\lin \mathbf{N}_{i+1},\
\vdash {\tt add}  : \mathbf{N}_i\lin \mathbf{N}_{j}\lin \mathbf{N}_{\max(i,j)+1}$ and
$\vdash {\tt mul}  : \mathbf{N}_i\lin !^{i}\mathbf{N}_{j}\lin \mathbf{N}_{i+j}
$.
From this we have for $\BSTA$ the following completeness for polynomials.
\begin{lemma}[\cite{GaboardiRonchi07csl}]
\label{polynomials}
Let $P$ be a polynomial and  $deg(P)$  its degree. 
Then there is a term $\tt P$ defining $P$
typable as:
 $$
\vdash {\tt P}:!^{deg(P)}\mathbf{N}\lin \mathbf{N}_{2deg(P)+1}
$$
\end{lemma}
Strings of booleans are represented by terms
of the shape 
$\lambda {\tt c}.\lambda {\tt z}.{\tt c}{\tt b}_0(\cdots({\tt c}{\tt b}_n{\tt z})\cdots)$
where ${\tt b}_i\in \{\0,\1\}$. Such terms are typable
by the indexed type $\mathbf{S}_i\df  \forall \alpha .!^{i}
(\bt\lin \alpha\lin\alpha)\lin\alpha\lin\alpha$. 
Again, we write $\mathbf{S}$ to mean $\mathbf{S}_1$.
Moreover, there is a term $\tt len \df \lambda {\tt c}.\lambda{\tt s}. {\tt c}(\lambda \x.\lambda \y.{\tt s}\y) $ typable as 
$\vdash {\tt len}:\mathbf{S}_i\lin \mathbf{N}_i$
that given a string of booleans returns its length.
Note that the data types defined above can be typed in $\BSTA$ by derivations
with degree $0$.
\subsection*{Boolean connectives}
It is worth noting that due to the presence of the $(\bt E)$ rule it is possible to 
define the usual boolean connectives.
Remembering that in our language $\0$ denotes ``true'' while 
$\1$ denotes ``false'', we have the following terms:
$$\M\ {\tt and}\ \N\df\lif{\M}{(\lif{\N}{\0}{\1})}{\1}$$
$$\M\ {\tt or}\ \N\df\lif{\M}{\0}{(\lif{\N}{\0}{\1})}$$ 
It is worth noticing that due to the presence of the $(\bt E)$ rule,
the following rules with an additive management of contexts are derivable in $\BSTA$:
$$
\infer{\Gamma\vdash \M\ {\tt and}\ \N:\bt }{\Gamma \vdash \M:\bt & \Gamma\vdash \N:\bt}
\qquad 
\infer{\Gamma\vdash \M\ {\tt or}\ \N:\bt }{\Gamma \vdash \M:\bt & \Gamma\vdash \N:\bt}
$$
Moreover, there is a term $\tt not$ defining the expected boolean function.
\subsection*{ATMs Configurations}
The encoding of Deterministic Turing Machine configuration given in \cite{GaboardiRonchi07csl}
can be adapted
in order to encode Alternating Turing Machine configurations.
In fact, an ATM  configuration can be viewed as a DTM 
configuration with an extra information about the state.
There are four kinds of state: 
\emph{accepting $(\tt A)$, rejecting $(\tt R)$, universal $(\land)$} and \emph{existential $(\lor)$ }. We can encode
such information by tensor pairs of booleans. In particular:
$$
\begin{tabular}{|c|c||c|c||c|c||c|c|}
\hline
$\langle \1,\0\rangle$ & $\tt A$ &$\langle \1,\1\rangle$ & $\tt R$ &
 $\langle \0,\1\rangle$ & $\land$ & $\langle \0,\0\rangle$ & $\lor$ \\
\hline
\end{tabular}
$$
We say that a configuration is accepting, rejecting, universal or  existential
depending on the kind of its state.\\
We can encode ATM configurations
by terms of the shape:
$$
{\tt  \lambda c}.\langle {\tt c}{\tt b}_0^l\circ\cdots \circ {\tt c}{\tt b}_n^l,{\tt c}{\tt b}_0^r\circ \cdots\circ {\tt c}{\tt b}_m^r,\langle {\tt Q},{\tt k}\rangle \rangle 
$$ 
where ${\tt c}{\tt b}_0^l\circ\ldots\circ {\tt c}{\tt b}_n^l$ and 
 ${\tt c}{\tt b}_0^r\circ \ldots\circ {\tt c}{\tt b}_n^r$ are respectively the
 left and right hand side words on the ATM tape,  
$\tt Q$ is a tuple of length $q$ encoding the state and  $\tt k\equiv \langle {\tt k}_1,{\tt k}_2\rangle$ is the tensor pair 
encoding the kind of the state. 
By convention, the left part of the tape is 
represented in a reversed order,
 the alphabet is composed by the two symbols $\0$ and
$\1$, the scanned symbol is the first symbol in the right part 
and final states are divided in accepting and rejecting.\\
Each term representing a configuration can be typed by indexed 
types (for every $i\geq 1$) as:
$$
\mathbf{ATM}_i\df\forall \alpha. !^i(\mathbf{B}\lin \alpha\lin\alpha)\lin ((\alpha\lin\alpha)^{2}\otimes\bt^{q+2} )
$$
We need some terms defining operations on ATM. 
In particular, the term 
${\tt Init} \df\lambda {\tt t} .\lambda {\tt c}. \langle \lambda {\tt z}.{\tt z}, \lambda {\tt z}. {\tt t}({\tt c}\mathbf{0}){\tt z},\langle \underline{{\tt Q}_0}, \underline{{\tt k}_0} \rangle \rangle$
defines the initialization function that takes in input a Church numeral $\underline{\tt n}$ and
gives as output a 
Turing machine with tape of length $n$ filled by $\mathbf{0}$'s
in the initial state
${\tt Q}_0\equiv \langle {\tt q}_0,\ldots,{\tt q}_n\rangle $ of kind
 ${\tt k}_0\equiv \langle {\tt k}_0',{\tt k}_0''\rangle$
and with the head at the 
beginning of the tape. It is easy to verify that 
${\tt Init}:  \mathbf{S}_i\lin\mathbf{ATM}_i$ for every $i\geq 1$.\\
An ATM transition relation $\delta$ can be considered as the union of
the transition functions $\delta_1,\ldots,\delta_n $ of its components. So, we need to
show that transition functions are definable.
We decompose an ATM transition step in two stages.
In the first stage, the ATM configuration is decomposed  
to extract the information needed by the transition relation. 
In the second one, the previously
obtained information are combined, 
depending on the considered transition function $\delta_j$, in order to build
the new ATM configuration. 
The term performing the decomposition stage is:
$$\begin{array}{l}
{\tt Dec}\df\lambda {\tt s}.\lambda {\tt c}.\letp{{\tt s}({\tt F[c]})}{{\tt l},{\tt r}}{{\tt p}}{} {\tt let}\  {\tt p}
\ {\tt be} \ {\tt q},{\tt k}\ {\tt in}\ 
 {\tt let}\  {\tt l}\langle {\tt I},\lambda \x.{\tt I}, \mathbf{0} \rangle\\
\quad {\tt be} \ {\tt t}_l,{\tt c}_l,{\tt b}_0^l\ {\tt in}\ 
\letp{{\tt r}\langle {\tt I},\lambda \x.{\tt I}, \mathbf{0} \rangle}{{\tt t}_r,{\tt c}_r}{{\tt b}_0^r}{}
\langle {\tt t}_l,{\tt t}_r,{\tt c}_l,{\tt b}_0^l,{\tt c}_r,{\tt b}_0^r ,{\tt q},{\tt k}\rangle 
\end{array}
$$
where 
$
\tt F[c]\df\lambda b.\lambda z.\letp{z}{g,h}{i}{
\langle hi\circ g,c,b  \rangle}
$.
It is boring but easy to check that the term $\tt Dec$ 
can be typed as $\vdash {\tt Dec}:\mathbf{ATM}_i \lin \mathbf{ID}_i$, where the indexed type $\mathbf{ID}_i$ is used to type the intermediate configuration decomposition and it is defined as
$%
\mathbf{ID}_i\df\forall \alpha. !^i(\mathbf{B}\lin \alpha\lin\alpha)\lin ((\alpha\lin\alpha)^2\otimes((\mathbf{B}\lin \alpha\lin\alpha) \otimes \mathbf{B})^2
\otimes \mathbf{B}^q\otimes \mathbf{B}^2)
$.
The behaviour of $\tt Dec$ is the following:
\begin{multline*}
{\tt Dec}\ (\lambda {\tt c}.\langle {\tt cb}_0^l\circ\cdots \circ {\tt cb}_n^l,{\tt cb}_0^r\circ\cdots \circ {\tt cb}_m^r, \langle {\tt Q}, {\tt k}\rangle \rangle )\to^{*}_{\beta}\\
\lambda {\tt c}.\langle {\tt cb}_1^l\circ\cdots \circ {\tt cb}_n^l , {\tt cb}_1^r\circ\cdots \circ {\tt cb}_m^r,{\tt c},{\tt b}_0^l,{\tt c},{\tt b}_0^r, {\tt Q}, {\tt k} \rangle
\end{multline*}
The transition combination stage is performed by the term 
$$\begin{array}{l}
{\tt Com}\df\lambda {\tt s}.\lambda{\tt c}.\letp{{\tt sc}}{{\tt l,r,c}_l,{\tt b}_l,{\tt c}_r,{\tt b}_r}{{\tt q},{\tt k}}{}\\
\quad\letp{\underline{\tt \delta}_j\langle {\tt b}_r,{\tt q},{\tt k}\rangle }{{\tt b',q',k'}}{{\tt m}}{}
 ({\tt if}\ {\tt m} \ {\tt then} \ {\tt R} \ {\tt else} \ 
{\tt L}) {\tt b'q'k'}
\langle {\tt l,r,c}_l,{\tt b}_l,{\tt c}_r\rangle
\end{array}
$$
where
${\tt R}\df
\lambda {\tt b'}.\lambda{\tt q'}.\lambda{\tt k'}.\lambda{\tt s}. \letp{{\tt s}}{{\tt l,r,c}_l,{\tt b}_l}{{\tt c}_r}{
\langle {\tt c}_r{\tt b'}\circ {\tt c}_l{\tt b}_l\circ {\tt l} ,{\tt r},\langle{\tt q}',{\tt k'}\rangle\rangle}$, 
${\tt L}\df\lambda {\tt b'}.\lambda{\tt q'}.\lambda{\tt k'}.\lambda{\tt s}. \letp{{\tt s}}{{\tt l,r,c}_l,{\tt b}_l}{{\tt c}_r}{
\langle {\tt l, c}_l{\tt b}_l\circ {\tt c}_r{\tt b'}\circ {\tt r},\langle{\tt q'}, {\tt k'}\rangle\rangle}$
and $\underline{\tt \delta}_j$ is a term defining the $\delta_j$
 component of the transition relation $\delta$.
The term $\tt Com$ 
can be typed as $\vdash {\tt Com}: \mathbf{ID}_i\lin\mathbf{ATM}_i$.
It  combines the symbols obtained after the decomposition stage depending on the considered component $\delta_j$ and 
returns the new ATM configuration. 
If $\delta_j({\tt b}_0^r,{\tt Q}, {\tt k})=({\tt b}',{\tt Q}', {\tt k}',\textrm{Right}) $, then 
\begin{multline*}
{\tt Com}\ (\lambda {\tt c}.\langle {\tt cb}_1^l\circ\cdots \circ {\tt cb}_n^l , {\tt cb}_1^r\circ\cdots \circ {\tt cb}_m^r,{\tt c},{\tt b}_0^l,{\tt c},{\tt b}_0^r, \langle {\tt Q}, {\tt k}\rangle \rangle)
\to^{*}_{\beta}\\
 \lambda {\tt c}.\langle {\tt cb}'\circ {\tt cb}_0^l\circ {\tt cb}_1^l\circ\cdots \circ {\tt cb}_n^l , {\tt cb}_1^r\circ\cdots \circ {\tt cb}_m^r,\langle {\tt Q}',{\tt k}'\rangle \rangle\
\end{multline*}
otherwise, if $\delta_j({\tt b}_0^r,{\tt Q}, {\tt k})=({\tt b}',{\tt Q}',{\tt k}',\textrm{Left}) $ then 
\begin{multline*}
{\tt Com}\ (\lambda {\tt c}.\langle {\tt cb}_1^l\circ\cdots \circ {\tt cb}_n^l , {\tt cb}_1^r\circ\cdots \circ {\tt cb}_m^r,{\tt c},{\tt b}_0^l,{\tt c},{\tt b}_0^r, \langle {\tt Q}, {\tt k}\rangle \rangle)
\to^{*}_{\beta}\\
\to^{*}_{\beta}\lambda {\tt c}. \langle  {\tt cb}_1^l\circ\cdots \circ {\tt cb}_n^l ,{\tt cb}_0^l\circ {\tt cb}'\circ {\tt cb}_1^r\circ\cdots \circ {\tt cb}_m^r,\langle {\tt Q}', {\tt k}'\rangle\rangle 
\end{multline*}
The term that takes a configuration and return its kind is:
$$\tt Kind\df\lambda \x.\letp{x(\lambda b.\lambda y.y)}{l,r}{s}{(\letp{s}{q}{k}{k})}$$ 
which is typable as $\vdash \tt Kind: \mathbf{ATM}_i\lin\bt^{2}$.
Finally the term
$$\tt Ext\df\lambda \x.\letp{({\tt Kind\ x})}{l}{r}{\tt r}$$ 
 typable as $\vdash \tt Ext: \mathbf{ATM}_i\lin\bt$,
returns $\0$ or $\1$ according to the fact that a given configuration is either accepting or rejecting.
\subsection*{Evaluation function}
Given an ATM $\mathcal{M}$ working in polynomial time 
we define a recursive evaluation procedure 
$\tt eval_{\mathcal{M}}$
that takes a string $\tt s$ and returns $\0$ or $1$ 
if the initial configuration (with the tape filled with $\tt s$) leads to 
an accepting or rejecting configuration respectively.\\
Without loss of generality we consider ATMs with transition relation $\delta$ of
degree two. So in particular, at each step we have two transitions terms ${\tt Tr}^1_{\mathcal{M}}$ and ${\tt Tr}^2_{\mathcal{M}}$ defining the two components $\delta_1$ and $\delta_2$ of the transition relation of $\mathcal{M}$.
We need to define some auxiliary functions.
In particular, we need a function $\alpha$ acting as
$$
\begin{array}{l}
\alpha( {\tt A}, \M_1,\M_2)={\tt A}\\
\alpha( {\tt R}, \M_1,\M_2)={\tt R}
\end{array}
\qquad
\begin{array}{l}
\alpha( \land, \M_1,\M_2)= \M_1\land\M_2\\
\alpha( \lor, \M_1,\M_2)= \M_1\lor\M_2\\
\end{array}
$$
This can be defined by the term
$$
\begin{array}{l}
\tt \alpha(M_0,M_1,M_2)\df \letp{M_0}{a_1}{a_2}{} 
\ifz a_1 \thenz
\tt  ( \ifz a_2 \thenz 
\langle a_1, \\
\tt \pi_2(M_1)\ or\ \pi_2(M_2)\rangle \elsez \langle a_1, \pi_2(M_1)\ and\ \pi_2(M_2)\rangle)
\elsez {\langle a_1, a_2\rangle}
\end{array}
$$
It is worth noting that $\alpha$ has typing:
$$
\begin{array}{l}
\infer{\Gamma \vdash  {\tt \alpha}(\M_0,\M_1,\M_2):\bt^{2}}{
\Gamma \vdash  \M_0:\bt^{2}&
\Gamma \vdash  \M_1:\bt^{2}&
\Gamma \vdash  \M_2:\bt^{2}
}
\end{array}
$$
where the contexts management is additive. This is one of the main reason for introducing  the $\ifz$ rule with an 
additive management of contexts. Moreover, note that we do not need 
any modality here, in particular this means that the $\alpha$ function  can be defined in the linear fragment of the $\BSTA$ system.
\\
The evaluation function $\tt eval_{\mathcal{M}}$ 
can now be defined as an iteration of an higher order
$\tt Step_{\mathcal{M}}$ function over a $\tt Base$ case. Let $\tt Tr^{1}_{\mathcal{M}}$ and
$\tt Tr^2_{\mathcal{M}}$ be two closed terms  defining the two components
of the transition relation. Let us define
$$
\begin{array}{l}
\tt Base \df\lambda c. (Kind\ c)\\
\tt Step_{\mathcal{M}}\df\lambda h.\lambda c.\alpha ((Kind\ c),
     ( h (Tr^1_{\mathcal{M}}\ c)), ( h ( Tr^2_{\mathcal{M}}\ c)))\\
\end{array}
$$
It is easy to verify that such terms are typable as:
$$
\begin{array}{l}
\vdash{\tt Base}:\TM_i\lin \bt^{2}\\
\vdash {\tt Step}_{\mathcal{M}}:(\TM_i\lin \bt^2) \lin\TM_i\lin \bt^2\\
\end{array}
$$
Let $P$ be a polynomial definable by a term $\tt P$ 
typable as 
$\vdash {\tt P}:!^{deg(P)}\mathbf{N}\lin \mathbf{N}_{2deg(P)+1}$. 
Then, the evaluation function of an ATM $\mathcal{M}$ working in polynomial time $P$ is
definable by the term:
$$\tt eval_{\mathcal{M}}\df \lambda s. Ext ((P\ (len\ s)\ Step_{\mathcal{M}} \ Base )(Init \ s)) 
$$
which is typable in $\BSTA$ as
$
\vdash {\tt eval}_{\mathcal{M}}:!^{t}\mathbf{S}\lin \bt$ where 
$t=\max(deg(P),1)+1
$.\\
Here, the evaluation is performed by a higher order iteration, which
represents a recurrence with parameter substitutions.
Note that by considering an ATM $\mathcal{M}$ that decides a language  $\mathcal{L}$, we have that 
the final configuration is either accepting or rejecting.  Hence
the term $\tt Ext$ can be applied with the intended meaning.
\begin{lemma}
A \emph{decision problem} $\mathcal{D}: \{0,1\}^{*}\to \{0,1\}$ 
decidable by an ATM $\mathcal{M}$ in polynomial time 
is $\bt$-programmable in $\BSTA$. 
\end{lemma}
\begin{proof}
$
  \mathcal{D}(s)=b \iff \tt eval_{\mathcal{M}} s\ev \tt b$ 
\end{proof}
From the well known result of \cite{ChKS81} we can conclude.
\begin{theorem}[Polynomial Space Completeness]
Every decision problem $\mathcal{D}\in\mathrm{\emph{PSPACE}}$ 
is $\bt$-programmable in $\BSTA$.
\end{theorem}

%% file: conclusion.tex
\section{Conclusion}
In this paper we have designed $\BSTA$, a language correct and complete with respect the polynomial space 
computations. Namely, the calculus is an extension of $\lambda$-calculus, and we supplied a type 
assignment system for it, such that well typed programs (closed terms of constant type) can be evaluated
in polynomial space and moreover all polynomial space decision functions can be computed by well 
typed programs. In order to perform the complexity bounded evaluation a suitable 
evaluation machine $\mathrm{K}_{\mathcal{B}}^{\mathcal{C}}$ has been defined, evaluating programs according to the left-most outer-most evaluation strategy and using two memory devices, one in order to make the evaluation space-efficient and the other in order to avoid backtracking.\par
\medskip 

The results presented in this paper have been obtained by exploiting
the equivalence ~\cite{ChKS81}:
$$
  \text{PSPACE} = \text{APTIME}
$$
Indeed, evaluations in the machine $\mathrm{K}_{\mathcal{B}}^{\mathcal{C}}$
can be regarded as computations in Alternating Turing Machines. Moreover, the simulation of big-step evaluations by means of 
small-step reductions is a reminiscence of the simulation 
of ATM by means of Deterministic Turing Machines.
Conversely, the PSPACE completeness is shown by encoding polynomial
time ATM by means of well typed terms. An interesting fact in the 
completeness proof is that the modal part 
of the $\BSTA$ system is only involved in the polynomial iteration, while
the ATM behaviour (i.e. the $\alpha$ function) can be defined in the 
modal free fragment of the system.
On the basis of these facts, we think that 
our tools could be fruitfully used in 
order to revisit some classical complexity results relating time and space \cite{DBLP:journals/tcs/Stockmeyer76}.\par
\medskip

Starting from the type system $\BSTA$ presented in this paper, one would wonder to exploit the proofs-as-programs correspondence in the design of a purely logical characterization of the class PSPACE. 
In particular, one would understand how to do this in sequent calculus
or proof nets,
the two proof formalisms most natural for linear logic.
Unfortunately, the logical sequent calculus system obtained by 
forgetting terms is 
unsatisfactory. Indeed, it looks not so easy to understand how to transfer 
the complexity bound from the term evaluation to the cut-elimination
in a logic.  
Moreover, boolean constants are redundant
and the $\BSTA$ rule $(\bt E)$ has no direct correspondent
in sequent calculus.  All these difficulties suggest that exploring this direction could be a true test for the light logics principles.


\medskip

\medskip  
